\documentclass[twocolumn]{autart}    % Enable this line and disable the 
\usepackage{graphicx}          % Include this line if your 
\usepackage[T1]{fontenc}
\usepackage[utf8]{inputenc}
\usepackage{mathptmx}
\usepackage{courier}
\usepackage{cite}
\usepackage{graphicx}
\usepackage[cmex10]{amsmath}
\usepackage{array}
\usepackage{color}
\usepackage{amssymb}
\usepackage{amsfonts}
\usepackage[short]{optidef}
\usepackage{balance}
\usepackage{subcaption}
\usepackage{float}
\usepackage{subfiles}
\usepackage{svg}
\usepackage{comment}
\usepackage{multirow, makecell}
\usepackage{multicol}
\usepackage{arydshln}
\usepackage{url}
\usepackage{nicematrix}
\usepackage{enumitem}

\definecolor{pyplotblue}{RGB}{0, 114, 189}
\definecolor{pyplotorange}{RGB}{255, 133, 27}
\definecolor{pyplotgreen}{RGB}{44, 159, 44}
\definecolor{pyplotyellow}{RGB}{237, 177, 32}
\definecolor{pyplotred}{RGB}{255, 0, 0}

\DeclareMathSymbol{\shortminus}{\mathbin}{AMSa}{"39}
\newcommand{\shortplus}{+}%{\raisebox{.4\height}{\scalebox{.6}{+}}}

\newtheorem{condition}[thm]{Condition}
\newtheorem{property}[thm]{Property}

\newcommand{\indeg}[1]{{\text{deg}^\shortminus}\left(#1\right)}
\newcommand{\outdeg}[1]{{\text{deg}^\shortplus}\left(#1\right)}

\newcommand{\base}{\text{base}}
\newcommand{\uk}{u_k}
\newcommand{\xbase}{x_\mathrm{b}}
\newcommand{\xbasek}{x_{\mathrm{b},k}}
\newcommand{\xbasekplus}{x_{\mathrm{b},k+1}}
\newcommand{\xaug}{x_\mathrm{a}}
\newcommand{\xaugk}{x_{\mathrm{a},k}}
\newcommand{\xaugkplus}{x_{\mathrm{a},k+1}}
\newcommand{\zbase}{z_\mathrm{b}}
\newcommand{\zbasek}{z_{\mathrm{b},k}}
\newcommand{\zaug}{z_\mathrm{a}}
\newcommand{\zaugk}{z_{\mathrm{a},k}}
\newcommand{\wbase}{w_\mathrm{b}}
\newcommand{\wbasek}{w_{\mathrm{b},k}}
\newcommand{\waug}{w_\mathrm{a}}
\newcommand{\waugk}{w_{\mathrm{a},k}}

\newcommand{\yestk}{\hat y_k}

\begin{document}

\begin{frontmatter}
                                              
\title{Learning-based augmentation of first-principle models: A linear fractional representation-based approach\thanksref{disclaimer}} % Title, preferably not more than 10 words.

\thanks[disclaimer]{This paper was not presented at any IFAC 
meeting. Implementation of the proposed method is available at \url{https://github.com/JanHHoekstra/Model-Augmentation-Public}\\
\hspace{8pt}This work is funded by the European Union (Horizon Europe, ERC, COMPLETE, 101075836) and has also been supported by %the European Defence Fund programme under grant agreement number No 101103386 and 
the Air Force Office of Scientific Research under award number FA8655-23-1-7061. Views and opinions expressed are however those of the authors only and do not necessarily reflect those of the European Union or the European Research Council Executive Agency% or the European Commission
. Neither the European Union nor the granting authority can be held responsible for them.}
% \thanks[code]{Implementation of the proposed method is available at \url{https://github.com/JanHHoekstra/Model-Augmentation-Public}}
%https://gitlab.com/janhhoeksra0/learning-based-model-augmentation

% \thanks[funding]{\hspace{8pt}This work is funded by the European Union (Horizon Europe, ERC, COMPLETE, 101075836) and has also been supported by the European Defence Fund programme under grant agreement number No 101103386 and has also been supported by the Air Force Office of Scientific Research under award number FA8655-23-1-7061. Views and opinions expressed are however those of the authors only and do not necessarily reflect those of the European Union or the European Research Council Executive Agency or the European Commission. Neither the European Union nor the granting authority can be held responsible for them.}

\author[TUE]{Jan H. Hoekstra%\thanksref{corresponding}
}\ead{j.h.hoekstra@tue.nl},    % Add the 
\author[HUN]{Bendeg\'uz M. Gy\"or\"ok}\ead{gyorokbende@sztaki.hun-ren.hu},
\author[TUE,HUN]{Roland T\'oth}\ead{r.toth@tue.nl},  % (ead) as shown
\author[TUE]{Maarten Schoukens}\ead{m.schoukens@tue.nl}  % (ead) as shown

\address[TUE]{Control Systems Group, Eindhoven University of Technology, The Netherlands}
\address[HUN]{Systems and Control Lab, HUN-REN Institute for Computer Science and Control, Budapest, Hungary}

%\thanks[corresponding]{Corresponding author.}
          
                             % keyword list or with the 
                                          % help of the Automatica 
                                          % keyword wizard

\begin{abstract} \vskip -3mm \noindent
    Nonlinear system identification %(NL-SI) 
    has proven to be effective in obtaining accurate models from data for %highly 
    complex real-world systems. In particular, recent encoder-based methods with \emph{artificial neural network state-space} (ANN-SS) models have achieved state-of-the-art performance on various benchmarks, using computationally efficient methods and offering consistent model estimation in the presence of noisy data. However, inclusion of prior knowledge of the system can be further exploited to increase (i) estimation speed, (ii) accuracy, and (iii) interpretability of the resulting models. This paper proposes a model augmentation method that incorporates prior knowledge from \emph{first-principles} (FP) models in a flexible manner. We introduce a novel \emph{linear-fractional-representation} (LFR) model structure that allows for the general representation of various augmentation structures including the ones that are commonly used in the literature, and an encoder-based identification algorithm for estimating the proposed structures together with appropriate initialisation methods. The performance and generalisation capabilities of the proposed method are demonstrated on the identification of a hardening mass-spring-damper system in a simulation study and on the data-driven modelling of the dynamics of an F1Tenth electric car using measured data. \vspace{-4mm}
\end{abstract}       % Abstract of not more than 200 words.

\end{frontmatter}

\setlength{\parskip}{6pt}

\vspace{-2mm}
\section{Introduction} \label{sec:Introduction}
% \OL{Need for system identification}
\vspace*{-7pt}
As control systems are becoming more complex and performance requirements surge, the need for accurate nonlinear models capable of efficiently capturing complicated behaviours of physical systems is rapidly increasing. It is common practice to derive baseline models using \emph{first-principle} (FP) methods, e.g., rigid body dynamics\cite{Schoukens2019}; however, these models provide only an approximate system description. Although more accurate FP models can be developed, this is a labour-intensive process, especially when additional physical effects—such as friction or aerodynamic forces—are included. Modelling these phenomena from first principles often requires dedicated experimental campaigns to identify and estimate the associated unknown parameters. Furthermore, the resulting models may become too complex to be handled analytically. In some cases, reliable FP descriptions of the to-be-modelled effects may not even exist, resulting in %limiting the modelling %process 
%to 
approximations with varying levels of fidelity.
% \vspace*{-6pt}

% \OL{Black-box identification delivers accurate models from data}
To overcome these issues, \emph{nonlinear system identification} (NL-SI) methods offer an alternative option to estimate models directly from measurement data \cite{Schoukens2019}. Black-box models, particularly those that incorporate \emph{artificial neural networks} (ANNs), have achieved unprecedented accuracy in capturing complex behaviours. In control applications, ANN-based \emph{state-space} (SS) models have proven to be effective in handling high-order systems and capturing complex nonlinear dynamics \cite{gerben2022}.
% \cite{suykens1995, Schoukens2021}. 
% \vspace{-6pt}

% \OL{Challenges with black-box identification}
% \OL{1. Lacking interpretability of estimated model}
Although black-box methods may result in accurate models, they also have serious downsides. First, flexible function approximators are difficult to interpret, leading to model behaviour that is not well understood. This in turn limits the reliability of the model during, e.g., extrapolation beyond the training data. This is a significant drawback for control applications, where interpretable models are preferred in the design process \cite{ljung2010perspectives, Panel2021}.
% \OL{2. Time spent on capturing expected behavior}
The second drawback is the significant time spent learning expected behaviour that has already been modelled thoroughly, e.g., FP-based understanding of the rigid-body-dynamics of the system.
% \vspace{-5pt}

% \OL{Model augmentation for system identification}
% Various methods in the literature aim to address these downsides, such as physics-informed neural networks \cite{raissi_physics-informed_2019} and physics-guided neural networks \cite{Daw2022}. 

Physics-informed neural networks \cite{raissi_physics-informed_2019} and physics-guided neural networks \cite{Daw2022} embed the prior knowledge of the physics in the form of equations (algebraic of partial differential) in the cost function, enforcing the learnt functions to fit to known physics behaviour. This leads to more interpretable models as the known physics are enforced, and shows faster learning convergence. These methods, however, require the knowledge of such physics equations and still rely on a black-box model to capture the entire system.
\vspace*{-1pt}

%for control applications
A promising approach is model augmentation, e.g., \cite{sun2020comprehensive, gotte2022composed, Groote2022, Shah2022}. This method combines baseline models with flexible function approximators, such as ANNs, in a combined model structure. As a result of this structural combination, the prior knowledge is directly captured in the baseline model and the learning components only need to model unknown dynamics. For control engineering, such a structure is beneficial, as it is clear how a well-understood baseline model is combined with  black-box elements. % resulting in the final model.
% For control engineering, this approach further enhances the interpretability of the obtained model, as it is clear how the well-understood baseline model is augmented and how this affects the behaviour of the resulting entire model.
\vspace*{-1pt}

% \OL{Problem in variety of model augmentation structures}
In the literature, there are a variety of different model augmentation structures, such as parallel\cite{sun2020comprehensive} and series\cite{gotte2022composed, Groote2022, Shah2022} interconnections. These interconnections reflect in what form the known baseline model is combined with the learning component that models the unknown behaviour of the system. Although different interconnections may result in an equally accurate model, model complexity and convergence speed also need to be considered, especially when the final model is utilised for control purposes.
%, such as online learning or robust control.
One interconnection may be equally accurate while having a less complex parameterisation compared to others. 
%Which interconnection this will be for a specific baseline model and system is not trivial to determine.
It is not trivial to determine which interconnection is the most advantageous for a specific baseline model and data-generating system, as the choice of the optimal interconnection depends on the unknown dynamics of the system. To address this, further research into model augmentation methods is required to develop, e.g., automatic model selection methods. To facilitate this research, a general model augmentation structure is required, such that model augmentation methods can be developed efficiently and compared across different works in literature. Such a general model augmentation structure is currently lacking in the literature.

% To address this, a search is required through the set of models spanned by the combination of the baseline model and the learning component. For such a search, a general model augmentation structure is needed such that all possible model augmentation structures can be represented in a general manner and make use of a common identification approach, instead of developing separate approaches for each model augmentation structure.

% In addition, such a common identification method is required. Neither such a general model augmentation structure nor a corresponding identification method have been considered so far in literature.
% For such a search, a general model augmentation structure is needed such that all possible model augmentation structures can be represented in a general manner. However, such a general model augmentation structure is lacking in the literature.
% \vspace{-6pt}

% \OL{Proposal unified model structure based on LFR}
%\TJ{The below section requires further explanation regarding the motivation for using LFR model structures. The current motivation based on robust control is too direct to the point, and the implied usefulness for robust control is too strong.}
To solve this problem, a general model augmentation structure based on a \emph{Linear Fractional Representation} (LFR) is proposed, which has been chosen for its modular and flexible nature, enabling a generalised form for augmenting the FP or the already known dynamics. The formulation of LFRs allows for systematic model augmentation while maintaining a clear separation between the baseline and learning components. The proposed model augmentation structure is able to express a wide range of model augmentation structures used in literature, and thus is a unified representation. Furthermore, LFRs are commonly used in the robust control field for uncertainty modelling in a generalised plant format. 
%Thus, in addition to the unified representation, the proposed structure then also has the advantage that it is compatible with well-established control algorithms for classical LFRs \cite{Zhou1996, Schoukens2018}, allowing the identified models to be applied directly in control design.
Thus, in addition to the general representation, the proposed structure also ensures compatibility with well-established control methodologies for classical LFRs \cite{Zhou1996}, making them a versatile choice for a wide range of applications. The price of the uniform model structure is that well-posedness problems may arise. To address this, we examine the computational graph of the proposed structure and provide well-posedness conditions based on this graph.
% provide a detailed theoretical discussion of the proposed methodology.
% \OL{Proposal identification algorithm}
We also propose an identification algorithm capable of handling the general LFR model augmentation structure with consistency guarantees, while also addressing the joint identification of both the baseline model and learning component parameters and managing overparameterisation through a physics-guided regularisation method.
% \vspace{-6pt}

The main contributions %of this work 
are summarised as% %follows
\footnote{A preliminary version of the methods discussed in this work was presented in \cite{hoekstra_learning-based_2024}. Compared to \cite{hoekstra_learning-based_2024}, the key differences are the full parameterisation of the LFR, the computational graph, a detailed %theoretical discussion of the 
well-posedness condition, an extended description of the identification algorithm, a more in-depth simulation study, and a real-world identification example using measured data.}:
\vspace*{-1pt}

\begin{itemize}
    \item We present a novel, general LFR-based model augmentation structure with two possible parameterisations: one fully parameterised, offering high flexibility in estimation at the cost of model interpretability, and the other structured and sparsely parameterised, promoting transparency, but requiring more prior knowledge to define the structure;
    \item We provide proof of the representation capability of the proposed LFR-based model augmentation structure, demonstrating that it is capable of representing all commonly used model augmentation structures in the existing literature;
    % \item We present a physics-guided regularization method to enhance model interpretability and extrapolation capabilities when the baseline parameters and the parameters of the learning component are estimated simultaneously;
    \item We provide %mathematical 
    conditions under which well-posedness of the proposed model structure is guaranteed;
    %\item We provide proof for consistent estimation of the model parameters in the proposed model structure;
    \item We provide an efficient identification algorithm for data-driven estimation of the augmented models under the proposed structure with consistency guarantees.
    \item We perform rigorous simulation and real-world measurement based studies to validate the proposed methodology.
\end{itemize}
\vspace*{-1pt}

The paper is organised as follows: first, the system identification problem and available model augmentation approaches are presented in Section~\ref{sec:ProblemStatement}. Then, Section~\ref{sec:Method} introduces the LFR-based model augmentation structure and its computational graph. Section~\ref{sec:well_posedness} provides conditions under which well-posedness of the model is proven. Next, Section~\ref{sec:Identification} provides an identification algorithm with consistency guarantees. A hardening \emph{mass-spring-damper} (MSD) simulation example and real-world experiments with an F1Tenth electric car are used to demonstrate the performance of the proposed identification method in Sections~\ref{sec:SimulationStudy} and \ref{sec:RealWorldIdent}, respectively. The conclusions are given in Section~\ref{sec:Conclusion}.

\vspace{-1mm}
\section{Model Augmentation Problem Setting} \label{sec:ProblemStatement}\vspace{-1mm}
% \OL{System description}
We consider the system to be identified given by the \emph{discrete-time} (DT) nonlinear representation
\vspace*{-6pt}
\begin{subequations}\label{eq:nl_dyn}
	\begin{align}
		x_{k+1}&=f(x_k,u_k),\\
		y_k&=h(x_k, u_k) +e_k,
	\end{align} %\vspace*{-10pt}
\end{subequations} 
\vskip -4mm
where $x_k \in \mathbb{R}^{n_\mathrm{x}}$ is the state, $u_k \in \mathbb{R}^{n_\mathrm{u}}$ is the input, $y_k \in \mathbb{R}^{n_\mathrm{y}}$ is the output signal of the system at time moment $k\in\mathbb{Z}$ with $e_k$ an i.i.d. white noise process with finite variance, representing measurement noise, $f: \mathbb{R}^{n_\mathrm{x}} \times \mathbb{R}^{n_\mathrm{u}} \rightarrow \mathbb{R}^{n_\mathrm{x}}$ is the state-transition function and $h: \mathbb{R}^{n_\mathrm{x}}\rightarrow \mathbb{R}^{n_\mathrm{y}}$ is the output function. This state-space representation is a general form that can describe a wide range of %system 
dynamics encountered in practice.

% \TJ{Keep what is needed for consistency proof.}

% \OL{Baseline model description in state-space form.}
We assume a baseline model of \eqref{eq:nl_dyn} is available in the %an NL-SS 
form
\vspace*{-6pt}
\begin{subequations}\label{eq:baseline_model}
    \begin{align}
        \xbasekplus & = f_\text{base}\left(\theta_\text{base}, \xbasek, \uk\right), \\
        \yestk & = h_\text{base}\left(\theta_\text{base}, \xbasek, \uk\right),
    \end{align}
\end{subequations}
\vspace{-22pt}

% \noindent
where $\xbasek \in \mathbb{R}^{n_{\xbase}}$ is the baseline model state, $\yestk \in \mathbb{R}^{n_\mathrm{y}}$ is the model output,  and $f_\text{base} : \mathbb{R}^{n_{\xbase}} \times \mathbb{R}^{n_\mathrm{u}} \rightarrow \mathbb{R}^{n_{\xbase}}$ with $h_\text{base}: \mathbb{R}^{n_{\xbase}} \times \mathbb{R}^{n_\mathrm{u}}  \rightarrow \mathbb{R}^{n_{\mathrm{y}}}$ are the baseline state-transition and output readout functions respectively, parameterised by $\theta_\text{base}\in\mathbb{R}^{n_{\theta_\text{base}}}$. The parameters $\theta_\base$ correspond to the physical parameters associated with system \eqref{eq:nl_dyn}. 

\begin{figure}
    \centering
    \includegraphics[width=1.0\linewidth]{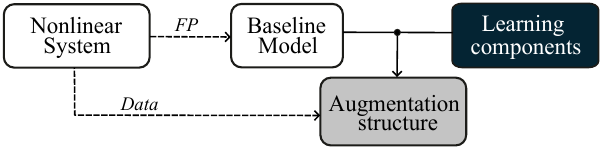}
    \caption{Learning-based model augmentation concept.}
    \label{fig:MA_concept}
\end{figure}

In model augmentation, the baseline model is combined with learning components in a combined model structure. The parameters of this model augmentation structure are then estimated using data measurements of the system as shown in Fig.~\ref{fig:MA_concept}. The general model augmentation structure is described as
\vspace*{-8pt}
\begin{subequations}\label{eq:augmentation_structure}
    \begin{align}
        \xbasekplus &= \left(f_\text{base} \star f_\text{aug}\right) \left(\xbasek, \xaugk, u_k\right), \\
        \xaugkplus &= g_\text{aug} \left(\xbasek, \xaugk, u_k\right), \\
        \hat{y}_k &= \left(h_\text{base} \star h_\text{aug}\right) \left(\xbasek, \xaugk, u_k\right),
    \end{align}
\end{subequations}
\vspace*{-24pt}

where $\xaugk$ are additional states added for dynamic augmentation structures, and $f_\text{aug}$, $g_\text{aug}$ and $h_\text{aug}$ are the learning components parameterised by $\theta_\text{aug} \in\mathbb{R}^{n_{\theta_\text{aug}}}$. For notational simplicity, both $\theta_\base$ and $\theta_\text{aug}$ are not written out in \eqref{eq:augmentation_structure}. The operator $\star$ represents an interconnection between two functions. This interconnection can represent a variety of different forms of model augmentation structure used in the literature, such as static parallel\cite{sun2020comprehensive} and static series~\cite{gotte2022composed, Groote2022, Shah2022} structures. Due to the state-space form of the model structure, augmentations can occur at the state and/or output level. We show a selection of possible state level augmentations in Table~\ref{tab:augmentation_structures} and output level augmentations in Table~\ref{tab:output_augmentation_structures}. Here, static refers to augmentations that do not add new state dimensions beyond the baseline model states $\xbase$. Dynamic augmentation structures\cite{bohlin2006practical, hoekstra_learning-based_2024}, on the other hand, add new augmentation states $\xaug$ to model missing dynamics. A broad range of further augmentations are possible. In this work, we restrict attention to the elementary augmentations listed in Tables \ref{tab:augmentation_structures} and \ref{tab:output_augmentation_structures}, which will be discussed through this paper.

\begin{table}
    \centering
    \caption{Classes of state model augmentation structures.}
    \begin{tabular}{l|l}
        \hline
        \multirow{2}{6em}{static parallel (S-SP)} & \multirow{2}{17em}{$\xbasekplus =  f_\text{base}\left(\xbasek, \uk\right) + f_\text{aug}\left(\xbasek, \uk\right)$}  \\
        & \\
        \hline
         \multirow{2}{7em}{static series output (S-SSO)} &  \multirow{2}{17em}{$\xbasekplus = f_\text{aug}\left(\xbasek, \uk, f_\text{base}\left(\xbasek, \uk\right)\right)$}  \\
         & \\
        \hline
        \multirow{2}{6em}{static series input (S-SSI)} &  \multirow{2}{17em}{$\xbasekplus = f_\text{base}\left(f_\text{aug}\left(\xbasek, \uk\right)\right)$}  \\
        & \\ 
        \hline
        \multirow{2}{7em}{dynamic parallel (S-DP)} &  $\xbasekplus = f_\text{base}\left(\xbasek, \uk\right) + f_\text{aug}\left(\xbasek, \xaugk, \uk\right)$ \\
         &  $\xaugkplus = g_\text{aug}\left(\xbasek, \xaugk, \uk\right)$ \\
        \hline
        \multirow{2}{7em}{dynamic series output (S-DSO)} &  $\xbasekplus = f_\text{aug}\left(\xbasek, \xaugk, \uk, f_\text{base}\left(\xbasek, \uk\right)\right)$  \\
        & $\xaugkplus = g_\text{aug}\left(\xbasek, \xaugk, \uk\right)$  \\
        \hline
        \multirow{2}{7em}{dynamic series input (S-DSI)} &  $\xbasekplus = f_\text{base}\left(f_\text{aug}\left(\xbasek, \xaugk, \uk\right)\right)$  \\
        & $\xaugkplus = g_\text{aug}\left(\xbasek, \xaugk, \uk\right)$  \\
        \hline
    \end{tabular}
    \label{tab:augmentation_structures}
\end{table}

\begin{table}
    \centering
    \caption{Classes of output model augmentation structures.}
    \begin{tabular}{l|l}
        \hline
        \multirow{2}{6em}{static parallel (O-SP)} & \multirow{2}{17em}{$\yestk = h_\text{base}\left(\xbasek, \uk\right) + h_\text{aug}\left(\xbasek, \uk\right)$}  \\
        & \\
        \hline
        \multirow{2}{7em}{static series output (O-SSP)} &  \multirow{2}{15em}{$\yestk = h_\text{aug}\left(\xbasek, \uk,h_\text{base}\left(\xbasek, \uk\right)\right)$}  \\
        & \\
        \hline
        \multirow{2}{6em}{static series input (O-SSI)} &  \multirow{2}{15em}{$\yestk = h_\text{base}\left(h_\text{aug}\left(\xbasek, \uk\right)\right)$}  \\
        & \\ 
        \hline
        \multirow{2}{7em}{dynamic parallel (O-DP)} &  $\yestk =  h_\text{base}\left(\xbasek, \uk\right) + h_\text{aug}\left(\xbasek, \xaugk, \uk\right)$  \\
        &  $\xaugkplus = g_\text{aug}\left(\xbasek, \xaugk, \uk\right)$  \\
        \hline
        \multirow{2}{7em}{dynamic series output (O-DSO)} &  $\yestk =  h_\text{aug}\left(\xbasek, \xaugk, \uk, h_\text{base}\left(\xbasek, \uk\right)\right)$ \\ 
        & $\xaugkplus = g_\text{aug}\left(\xbasek, \xaugk, \uk\right)$ \\
        \hline
        \multirow{2}{7em}{dynamic series input (O-DSI)} &  $\yestk = h_\text{base}\left(h_\text{aug}\left(\xbasek, \xaugk, \uk\right)\right)$ \\ 
        & $\xaugkplus = g_\text{aug}\left(\xbasek, \xaugk, \uk\right)$ \\
        \hline
    \end{tabular}
    \label{tab:output_augmentation_structures}
\end{table}
%
% The augmentation structures shown in Tables~\ref{tab:augmentation_structures}~and~\ref{tab:output_augmentation_structures} are only a small portion of possible interconnections between a given baseline model and the learning components as described in \eqref{eq:augmentation_structure}. The choice of a model augmentation structure represented by such a specific interconnection is an important choice.
%
% \vspace{-5mm}
As discussed in the introduction, a general augmentation structure is desired. For this, a parameterisation of the operator $\star$ is required, capable of characterising the interconnection between the baseline model and the learning components. This would realise a fully parameterised general augmentation structure of \eqref{eq:augmentation_structure}. Additionally, an identification algorithm able to estimate the parameters of this general model augmentation structure is proposed, under the restriction that the model is well-posed.

\section{LFR-based augmentation structure} \label{sec:Method}
\vspace{-1mm}
In this section, we formulate a general representation of \eqref{eq:augmentation_structure} in an LFR-based augmentation structure. Next, we derive the graph based representation of the proposed model structure, which is to be used to introduce sparsity to the LFR-based augmentation structure as well as for deriving conditions for well-posedness in Section~ \ref{sec:well_posedness}. Finally, we introduce the parameterisation of the learning component that we will use to formulate our augmentation approach.

% Second, we extend the LFR-based structure to a structured LFR-based augmentation structure based on graph theory. Based on this structured LFR-based structure, we then pose well-posedness conditions for both introduced model augmentation structures.

% Finally, we provide a conversion between the two representations and provide conditions for well-posedness.

% In this section, we consider two representations of the general model augmentation structure. First, we formulate a general representation of \eqref{eq:augmentation_structure} in a graph setting that defines the interconnection as an adjacency matrix. Second, we extend the interconnection structure to an LFR structure with a fully parameterised interconnection matrix that can itself represent the linear dynamics of the system. Finally, we provide a conversion between the two representations and provide conditions for well-posedness.

\subsection{General LFR-based augmentation structure}\label{sec:parametrised_interconnection}\vspace{-1mm}
% Notably, the interconnection augmentation structure \eqref{eq:interconnection} closely resembles the well-known LFR formulation. 
%
% The flexibility of the LFR representation has made it widely used in robust control \cite{Zhou1996, junnarkar2022synthesis} and linear parameter-varying control \cite{Toth2010, Schoukens2018}. 
%
As discussed in Section~\ref{sec:ProblemStatement}, many model augmentation structures are available in the literature, and now we propose a unified structure based on the \emph{Linear Fractional Representation} (LFR) that can represent all augmentation arrangements. 
%A key property of LFRs is that the primitive operations of summation, multiplication, and inversion allow manipulations with LFRs, and an LFR of an LFR is also an LFR\cite{redheffer1960certain}.
The flexibility of this representation has made it popular in the field of robust control \cite{Zhou1996} and linear parameter-varying-control \cite{Toth2010}. Furthermore, an LFR can also include nonlinear components in the interconnections\cite{veenman2016robust}, which has made LFRs useful for black-box nonlinear system representations\cite{Schoukens2020, shakib2022computationally} and implicit learning\cite{el2021implicit}. Recent results have also shown that stability properties can be enforced on these black-box LFRs in a constraint-free manner at the cost of some representation capability \cite{revay2020contracting, frank2022robust}.

%We begin by 
Introduce the following notation for the baseline terms:
\vspace*{-10pt}
\begin{equation}
    \phi_\mathrm{base}(\theta_\mathrm{base}, \zbasek) = \begin{bmatrix}
        f_\mathrm{base}(\theta_\mathrm{base},\zbasek)\\
        h_\mathrm{base}(\theta_\mathrm{base},\zbasek)
    \end{bmatrix}.
\end{equation}
\vspace*{-18pt}

Moreover, we denote the learning component as $\phi_\mathrm{aug}$, which can be represented by any universal function approximator. We assume that it is implemented as a function with $\theta_\mathrm{aug}\in\mathbb{R}^{n_{\theta_\mathrm{aug}}}$ collecting its parameters. Latent variables $\wbasek\in\mathbb{R}^{n_{\xbase}+n_y}$, $\waugk\in\mathbb{R}^{n_{\waug}}$, $\zaugk\in\mathbb{R}^{n_{\zaug}}$, and $\zbasek\in\mathbb{R}^{n_{\xbase}+n_u}$ are introduced, and expressed as \vspace{-1.5mm}
\begin{equation}\label{eq:LFR_zb_za}
    \begin{bmatrix}
        \zbasek\\
        \zaugk
    \end{bmatrix} = \begin{bmatrix}
        C_z^\mathrm{b}\\
        C_z^\mathrm{a}
    \end{bmatrix} \hat{x}_k + \begin{bmatrix}
        D_{zu}^\mathrm{b}\\D_{zu}^\mathrm{a}
    \end{bmatrix} \uk + \underbrace{\begin{bmatrix}
        D_{zw}^\mathrm{bb} & D_{zw}^\mathrm{ba}\\
        D_{zw}^\mathrm{ab} & D_{zw}^\mathrm{aa}
    \end{bmatrix}}_{D_{zw}} \begin{bmatrix}
        \wbasek\\ \waugk
    \end{bmatrix},
\end{equation}
\vspace*{-20pt}

where $C_z^\mathrm{b}$, $C_z^\mathrm{a}$, \dots, $D_{zw}^\mathrm{a}$ are real matrices with dimensions compatible with the signal dimensions, and their elements are parameters that are optimised during model learning. The state transition and output equations are expressed as\vspace{-1.5mm}
\begin{align*}\label{eqs:state_output_LFR}
    \underbrace{\begin{bmatrix}
        \xbasekplus\\ \xaugkplus
    \end{bmatrix}}_{\hat{x}_{k+1}} & = \underbrace{\begin{bmatrix}
        A^\mathrm{bb} & A^\mathrm{ba}\\
        A^\mathrm{ab} & A^\mathrm{aa}
    \end{bmatrix}}_{A} \underbrace{\begin{bmatrix}
        \xbasek\\ \xaugk
    \end{bmatrix}}_{\hat{x}_k} + \underbrace{\begin{bmatrix}
        B_u^\mathrm{b}\\ B_u^\mathrm{a}
    \end{bmatrix}}_{B_u} \uk +\underbrace{\begin{bmatrix}
        B_w^\mathrm{bb} & B_w^\mathrm{ba}\\
        B_w^\mathrm{ab} & B_w^\mathrm{aa}
    \end{bmatrix}}_{\left[\begin{array}{c:c}B_w^\mathrm{b} & B_w^\mathrm{a}\end{array}\right]} \underbrace{\begin{bmatrix}
        \wbasek\\ \waugk
    \end{bmatrix}}_{w_k},\\
    \hat{y}_k & = \underbrace{\begin{bmatrix}
        C_y^\mathrm{b} & C_y^\mathrm{a}
    \end{bmatrix}}_{C_y} \hat{x}_k + D_{yu} \uk + \begin{bmatrix}
        D_{yw}^\mathrm{b} & D_{yw}^\mathrm{a}
    \end{bmatrix} w_k,
\end{align*}
\vspace*{-20pt}

where $A$, $B_u$, $C_y$, and $D_{yu}$ are real matrices with appropirate signal dimensions representing the linear parts of the unmodeled dynamics, the baseline part participates in the relation through matrices $B_w^\mathrm{b}$ and $D_{yw}^\mathrm{b}$, while the nonlinear black-box terms effect the dynamics through  $B_w^\mathrm{a}$ and $D_{yw}^\mathrm{a}$.

Finally, the LFR-based model augmentation structure (shown in Fig.~\ref{fig:LFR-based_structure}) can be expressed in a compact form, as \vspace{-1.5mm}
\begin{subequations}\label{eqs:general_LFR}
\begin{align}
    \left[\begin{array}{c}
        \hat{x}_{k+1}\\ \hat{y}_k \\ \hdashline[5pt/2pt]  \zbasek\\ \zaugk
    \end{array}\right] &= \underbrace{\left[\begin{array}{cc;{5pt/2pt}cc}
        A & B_u & B_w^\mathrm{b} & B_w^\mathrm{a}\\
        C_y & D_{yu} & D_{yw}^\mathrm{b} & D_{yw}^\mathrm{a}\\ \hdashline[5pt/2pt] 
        C_z^\mathrm{b} & D_{zu}^\mathrm{b} & D_{zw}^\mathrm{bb} & D_{zw}^\mathrm{ba}\\
        C_z^\mathrm{a} & D_{zu}^\mathrm{a} & D_{zw}^\mathrm{ab} & D_{zw}^\mathrm{aa}
    \end{array}\right]}_{W(\theta_\mathrm{LFR})} \begin{bmatrix}
        \hat{x}_k\\ u_k\\ \hdashline[5pt/2pt]  \wbasek\\ \waugk
    \end{bmatrix}, \\
        \wbasek & = \phi_\text{base}\left(\theta_\text{base}, \zbasek\right), \label{eq:LFR_zb}\\
        \waugk & = \phi_\text{aug}\left(\theta_\text{aug}, \zaugk\right), \label{eq:LFR_za}
\end{align}
\end{subequations}
where $W$ is the LFR matrix, and all parameters determining the matrices $A$, $B_u$, \dots, $D_{zw}^\mathrm{aa}$ are collected into $\theta_\mathrm{LFR}$. Since $\theta_\mathrm{LFR}$ is included in the (tunable) model parameters, the final interconnection structure of the LFR-based augmentation is formed throughout model learning, hence the flexibility of the approach. 
However, an inherent challenge of LFR model structures is ensuring well-posedness of the model structure. As the model structure allows algebraic loops, it is possible the retrieve ill-posed realisations. We define the \emph{well-posedness} (WP) property as
% The feedback connection inherent in the proposed model augmentation structure requires an analysis of the well-posedness of the model structure and its parameterisation, as this is essential to ensure well-posedness during the estimation process. First, the \emph{well-posedness} (WP) property is defined as
%

\begin{figure}
    \centering
    \includegraphics[width=0.5\textwidth]{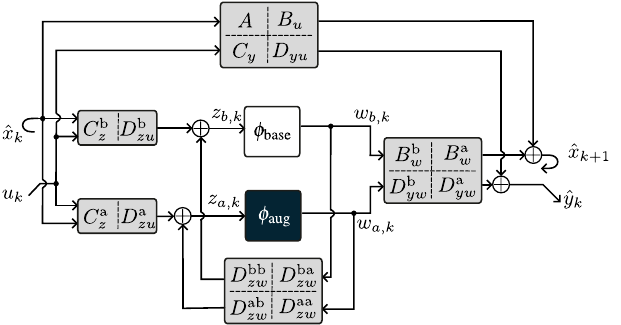}
    \caption{LFR-based augmentation structure of the baseline model characterised by $\phi_\text{base}$ with learning components $\phi_\text{aug}$. \vspace*{-6pt}}
    \label{fig:LFR-based_structure}
\end{figure}

\begin{property}[Well-posedness]\label{condition:well-posedness}
For any value of $\hat x_k \in \mathbb{R}^{n_\mathrm{x}}$ and $u_k \in \mathbb{R}^{n_\mathrm{u}}$, the signal relations in \eqref{eqs:general_LFR} admit a unique solution $z_k = \begin{bmatrix} \zbasek^\top & \zaugk^\top \end{bmatrix}^\top$.
\end{property}
Conditions to ensure the well-posedness are introduced in Section~\ref{sec:well_posedness}.
The unified representation capability of the proposed LFR structure is shown by the following theorem.
% As discussed in Section~\ref{sec:ProblemStatement}, the LFR-based augmentation structure aims to represent a wide range of model augmentation representations, some of which are already present in the literature. We show this representation capability by the following theorem.
%
% These representations, such as the ones in Tables~\ref{tab:augmentation_structures}~and~\ref{tab:output_augmentation_structures}, can be achieved by assigning specific values to the LFR matrix $W$, as shown in Appendix~\ref{appendix:general_LFR_augmentation_forms}.
%
\begin{thm}[Unified representation]
\label{thm:representation_ma_structures}
Given a baseline model \eqref{eq:baseline_model} with parameters $\theta_\mathrm{base}$ connected to learning components $f_\mathrm{aug}$, $g_\mathrm{aug}$ and $h_\mathrm{aug}$, each parametrised with $\theta_\mathrm{aug}$, in terms of \eqref{eq:augmentation_structure}, where the operator $\star$ corresponds to any of the interconnections listed in Table \ref{tab:augmentation_structures} and \ref{tab:output_augmentation_structures}. Then, for the considered model structure \eqref{eqs:general_LFR}, there exists a $\theta_\mathrm{LFR} \in \mathbb{R}^{n_{\theta_\text{LFR}}}$ and choices of latent dimensions  $n_{\waug}, n_{\zaug}  \geq 0$ and a $\phi_\mathrm{aug}$ function, such that \eqref{eq:augmentation_structure} and \eqref{eqs:general_LFR} are equivalent representations of the same dynamic behaviour.
% Let $f_{\text{base}}$ and $h_{\text{base}}$ denote the baseline model, 
% and let $f_{\text{aug}}$ denote the considered learning functions. 
% Then, there exists a parameterisation $\theta_\text{LFR} \in \mathbb{R}^{n_{\theta_\text{LFR}}}$ of the LFR matrix for each model augmentation structure 
% listed in Tables~\ref{tab:augmentation_structures} and~\ref{tab:output_augmentation_structures}, 
% such that the LFR-based model augmentation structure \eqref{eqs:general_LFR} 
% is equivalent to the corresponding model augmentation structure.
\end{thm}\vspace{-12pt}
\begin{pf}
    See Appendix~\ref{appendix:general_LFR_augmentation_forms}. \hfill $\blacksquare$
\end{pf}\vspace{-6pt}
% This theorem is proven by direct manipulation of the LFR matrix to show equivalence for the model augmentation structures; see Appendix~\ref{appendix:general_LFR_augmentation_forms}.
%
% For the model augmentation structures in Tables~\ref{tab:augmentation_structures}~and~\ref{tab:output_augmentation_structures}, we work out the equivalent interconnection representations in Appendix~\ref{appendix:Interconnection_representations}. Note that, this interconnection augmentation can represent all interconnections between $\phi_\text{base}$ and $\phi_\text{aug}$ as formulated by \eqref{eq:augmentation_structure}.
%
\vspace{-2mm}
\subsection{Computational graph of the interconnection}\label{sec:fixed_interconnection} \vspace{-1mm}
The full parameterisation of $W$ in the proposed LFR-based model augmentation structure gives a general representation of \eqref{eq:augmentation_structure}. However, due to full parameterisation, the actual interconnection of the learning and baseline model components is represented in a black-box fashion compared to the cases listed in Tables~\ref{tab:augmentation_structures} and \ref{tab:output_augmentation_structures}.

% \OL{Motivation in three factors:
% \begin{enumerate}
%     \item Sparsity in interconnection structure
%     \item Well-posedness proof
%     \item Initialisation algorithm
% \end{enumerate}
% }

To be able to detect or even enforce a particular configuration of these components, we investigate the computational graph of the augmentation interconnection \eqref{eq:augmentation_structure}. This graph makes clear what each edge, i.e., element of $\theta_\text{LFR}$, does in the augmentation structure, and thus also what removing it for sparsification will do. Additionally, the graph representation will provide general conditions on the well-posedness of the proposed model structure in Section~\ref{sec:well_posedness}.

First, we introduce computational graphs, then present the graph representations of the baseline model and the learning component, which are finally combined into the LFR model graph. Here we make use of the following graph notions: $G = (V,E)$ is a graph with nodes $V$ and edges $E$, $\text{deg}^-$ is the indegree of a node, $\text{deg}^+$ is the outdegree of a node, $\text{disjoint}\left(V,W\right)$ is the disjoint union of set between two sets $V$ and $W$ defined as $V \cup W$ with $V \cap W = \emptyset$, $\mathrm{source}\left(W\right)$ defines a set of nodes $V = \left\{v \in W | \indeg{v} = 0\right\}$, and $\mathrm{sink}\left(W\right)$ defines a set of nodes $V = \left\{v \in W | \outdeg{v} = 0\right\}$. We also introduce the shorthand notation $\mathrm{outvar}\left(V, W\right)$ to define the set of edges $E_{vw} = \left\{\left(v,w\right)| v \in V, w\in W, \outdeg{v} = \indeg{w} = 1 \right\}$, and $G_c = (V_c, E_c) = G \setminus V_s$ for vertex contraction with $G = (V, E)$, $V_s \subset V$, $V_c = V \setminus V_s \cup v_c$ and $v_c$ being the contracted vertex.
\begin{figure}
    \centering
    \begin{subfigure}{0.5\textwidth}
        \centering
        \includegraphics[width=.8\textwidth]{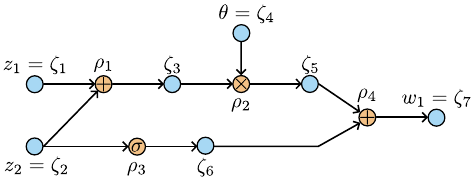} \vspace{-1mm}
        \caption{Full computational graph}
        \label{fig:comp_graph_full}
    \end{subfigure}%
    \vskip\baselineskip
    
    \begin{subfigure}{0.5\textwidth}
        \centering
        \includegraphics[width=0.4925\textwidth]{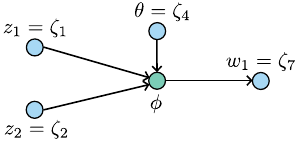} \vspace{-1mm}
        \caption{Computational graph with subgraph $\phi$}
        \label{fig:comp_subgraph}
    \end{subfigure}
    \vspace*{-18pt}
    \caption{Computational graph representation of the function $w_1 = \phi\left(\theta, z_1, z_2\right) = \left(z_1 + z_2\right)\theta + \sigma\left(z_2\right)$.}
    \label{fig:comp_graph} \vspace{-1mm}
\end{figure}

\subsubsection{Computational graphs}
We take the formulation of the \emph{computational graph} as defined in \cite{bauer1974computational}. First, take the \emph{computational problem} defined as the set of $M$ functions \vspace{-1mm}
\begin{equation}\label{eq:computational_problem}
    w_i = \phi_i = \left(z_1, \ldots ,z_K\right), \quad i=1,\ldots,M, \vspace{-1mm}
\end{equation}
of $K$  real variables $z_1, \ldots, z_K \in \mathbb{R}$ from which $M$ real quantities $w_i \in \mathbb{R}$ are obtained. Each function $\phi_i: \mathbb{R}^K \to \mathbb{R}$ can be described by a \emph{computational process} consisting of a number of primitive operators \vspace{-1mm}
\begin{equation}\label{eq:computational_process}
    \zeta_{i_{c,0}} := \rho_c\left(\zeta_{i_{c,1}}, \ldots, \zeta_{i_{c,s}}\right), \quad c = 1, \ldots, L \vspace{-1mm}
\end{equation}
where the values $\zeta_{i_{c,j}} \in \mathbb{R}$ and $\rho_c: \mathbb{R}^s \to \mathbb{R}$. For all $j\in[1, \ldots, s]$, $\zeta_{i_{c,j}} \in [z_1, \ldots, z_K,\zeta_{i_{b,0}}]$ with $b < c$, i.e., the inputs to $\rho_c$ can be the input of $\phi_i$ or output from a previous operator $\rho_b$. For $j=0$, $\zeta_{i_{c,j}} \in [w_1, \ldots, w_M,\zeta_{i_{b,1},\ldots,\zeta_{i_{b,s}}}]$ with $b > c$, i.e., the output of $\rho_c$ can be the output of $\phi_i$ or input to a next operator $\rho_b$.
% The output of an operator $\rho_a$ can be the input of the next operator $\rho_b$, e.g., $\zeta_{i_{a,0}} = \zeta_{i_{b,1}}$, and a value $\zeta_{i_{c,j}}$ can be input to multiple operators $\rho_c$.
% We consider primitive operators $\rho_c$ to be addition, multiplication, or a nonlinear one-to-one function. Thus, an operator $\rho_c$ can have one or two inputs, e.g., $s = 1 \vee 2$.
% We also note that $w_i$ is represented by a $\zeta_{i_{c,0}}$ and $z_1, \ldots, z_K$ are represented by some $\zeta_{i_{c,j}}$.
%
We consider primitive operators $\rho_c : \mathbb{R}^{s} \to \mathbb{R}$, with 
$s \in \{1,2\}$, to be either: addition, multiplication or a nonlinear injective function (where $s = 1$).
% \begin{itemize}
%     \item addition,
%     \item multiplication, or
%     \item a nonlinear injective function (where $s = 1$).
% \end{itemize}

The \emph{computational process} defines and it is characterised by a \emph{computational graph}, where both $\rho_c$ and $\zeta_{i,_{c,j}}$ are nodes in the graph. This computational graph is defined as follows:

% This \emph{computational graph} is a directed graph where the values $\zeta$ are nodes $V_\zeta$ with in-degree of at most 1, and each function $\rho$ of $s$ arguments corresponds to one node $V_\rho$ with in-degree $s$ and out-degree of 1 connecting to a node $V_\zeta$. The edges from $V_\rho$ to $V_\zeta$ are defined as
% \begin{equation}
%     E_{\rho\zeta} = \left\{\left(\rho,\zeta\right) | \rho\in V_\rho, \zeta \in V_\zeta \text{ and } \right\}
% \end{equation}
% and the edges from $V_\zeta$ to $V_\rho$ are defined as
% \begin{equation}
%     E_{\zeta\rho} = \left\{\left(\zeta,\rho\right) | \zeta \in V_\zeta, \rho\in V_\rho \text{ and } \frac{\rho_c}{\partial \zeta_j} \neq 0\right\}
% \end{equation}

\begin{defn}[Computational graph]\label{definition:computational_graph}
A directed graph denoted by $G = \left(V,E\right)$ with $V=\text{disjoint}\left(V_\zeta,V_\rho\right)$ the set of vertices and $E = \text{disjoint}\left(E_{\zeta\rho},E_{\rho\zeta}\right)$ the set of edges is called a computational graph if %such that 
\vspace{-1mm}
% \begin{enumerate}[label=(\alph*)]
%     \item\label{th:first} $V_\rho = \{\rho_1, \ldots, \rho_L\}$
%     \item\label{th:second} $V_\zeta = \{\zeta_1, \ldots, \zeta_K\}$
%     \item\label{th:source} $V_z = \left\{z \in V_\zeta | \indeg{z} = 0\right\}$
%     \item\label{th:sink} $V_w = \left\{w \in V_\zeta | \outdeg{w} = 0\right\}$
%     \item\label{th:inputs} $E_{\zeta\rho} = \left\{\left(\rho,\zeta\right)| \rho \in V_\rho, \zeta\in V_\zeta \text{ and } \frac{\partial \rho}{\partial \zeta} \neq 0\right\}$
%     \item\label{th:function_map} $E_{\rho\zeta} = \left\{\left(\rho,\zeta\right)| \rho \in V_\rho, \zeta\in V_\zeta \setminus V_z, \outdeg{\rho} = \indeg{\zeta} = 1 \right\}$
%     % \item \textcolor{red}{$\forall \rho \in V_\rho, \exists! \zeta \in V_\zeta \setminus V_z \text{ s.t. } \left(\rho, \zeta\right) \in E_{\rho \zeta} \enspace \Rightarrow \outdeg{\rho} = 1$.}
% \end{enumerate}
\begin{enumerate}[label=(\alph*)]
    \item\label{th:first} $V_\rho = \{\rho_1, \ldots, \rho_L\}$, $V_\zeta = \{\zeta_1, \ldots, \zeta_K\}$
    % \item\label{th:second} $V_\zeta = \{\zeta_1, \ldots, \zeta_K\}$
    \item\label{th:source} $V_z = \mathrm{source}\left(V_\zeta\right)$, $V_w = \mathrm{sink}\left(V_\zeta\right)$
    % \item\label{th:source} $V_z = \left\{z \in V_\zeta | \indeg{z} = 0\right\}$, $V_w = \left\{w \in V_\zeta | \outdeg{w} = 0\right\}$
    % \item\label{th:source} $V_z = \left\{z \in V_\zeta | \indeg{z} = 0\right\}$
    % \item\label{th:sink} $V_w = \left\{w \in V_\zeta | \outdeg{w} = 0\right\}$
    \item\label{th:inputs} $E_{\zeta\rho} = \left\{\left(\rho,\zeta\right)| \rho \in V_\rho, \zeta\in V_\zeta \text{ and } \frac{\partial \rho}{\partial \zeta} \neq 0\right\}$
    % \item\label{th:function_map} $E_{\rho\zeta} = \left\{\left(\rho,\zeta\right)| \rho \in V_\rho, \zeta\in V_\zeta \setminus V_z, \outdeg{\rho} = \indeg{\zeta} = 1 \right\}$
    \item\label{th:function_map} $E_{\rho\zeta} = \mathrm{outvar}\left(V_\rho, V_\zeta \setminus V_z \right)$
    % \item \textcolor{red}{$\forall \rho \in V_\rho, \exists! \zeta \in V_\zeta \setminus V_z \text{ s.t. } \left(\rho, \zeta\right) \in E_{\rho \zeta} \enspace \Rightarrow \outdeg{\rho} = 1$.}
\end{enumerate}
\end{defn}
%
% \ref{th:first} and \ref{th:second} define the sets of operator and variable nodes respectively. \ref{th:source} and \ref{th:sink} define subsets of the variables representing the inputs and outputs respectively (i.e., sources and sinks in graph nomenclature). Lastly, \ref{th:inputs} defines what variable nodes are inputs to which operators, and \ref{th:function_map} defines the connection of the output variables for each operator. In Fig.~\ref{fig:comp_graph_full} we show the computational graph for the example function $w_1 = \phi\left(\theta, z_1, z_2\right) = \left(z_1 + z_2\right)\theta + \sigma\left(z_2\right)$.
\ref{th:first} defines the sets of operator and variable nodes respectively. \ref{th:source} defines subsets of the variables representing the inputs and outputs respectively (i.e., sources and sinks in graph nomenclature). Lastly, \ref{th:inputs} defines what variable nodes are inputs to which operators, and \ref{th:function_map} defines the connection of the output variables for each operator. In Fig.~\ref{fig:comp_graph_full} we show the computational graph for the example function $w_1 = \phi\left(\theta, z_1, z_2\right) = \left(z_1 + z_2\right)\theta + \sigma\left(z_2\right)$.

By defining a subgraph $G_{\phi_i}$, we can formulate a more compact notation. The computational subgraph $G_{\phi_i}$, contains the nodes $V_{\phi_i} =  \left(V_{\zeta_i} \cup V_{\rho_i} \right) \setminus V_{w_i} \setminus V_{z_i}$ and the edges internal to these nodes. Then by vertex contraction $G \setminus V_{\phi_i}$ we get the contracted noted $v_{\phi_i}$ that represents the multivariate function $\phi_i$ in \eqref{eq:computational_problem}. Such a contracted node $v_{\phi_i}$ is shown for the example function in Fig.~\ref{fig:comp_subgraph}. 
% The compacted computational graph is now described as $G_{\phi_i} = \left(V_{\phi_i}, E_{\phi_i}\right)$, with $V_{\phi_i} =\text{disjoint}\left(V_w, V_z, \phi_i\right)$ and $E_{\phi_i}=\text{disjoint}\left(E_{z \phi_i}, E_{\phi_i w}\right)$
% This graph will have the nodes $V_z$, $V_w$ and the new node $\phi$. It will have the following edges
% \begin{equation}
%     E_{\phi w} = \left\{\left(\rho,\zeta\right)| \rho \in V_\rho, \zeta\in V_\zeta \setminus V_z, \outdeg{\rho} = \indeg{\zeta} = 1 \right\}
% \end{equation}
% \begin{equation}
%     E_{z \phi} = \left\{\left(\rho,\zeta\right)| \rho \in V_\rho, \zeta\in V_\zeta \setminus V_z, \outdeg{\rho} = \indeg{\zeta} = 1 \right\}
% \end{equation}
For each function $\phi_i$ in the computational problem, a computational graph $G_{\phi_i}$ can be defined. By taking the union of these graphs as
\begin{equation}\label{eq:graph_union}
    G_\phi = \left(V_{\phi}, E_{\phi}\right) = \left(V_{\phi_1} \cup \ldots \cup V_{\phi_M}, E_{\phi_1} \cup \ldots \cup E_{\phi_M}\right)
\end{equation}
we obtain the computational graph $G_\phi$ of the entire computational problem $\eqref{eq:computational_problem}$ with $K$ inputs and $M$ outputs, as visualised in Fig.~\ref{fig:comp_graph_Phi}. Note that due to the union, %describes that the 
nodes and edges may be shared between the computational graphs $G_{\phi_i}$.

\vspace{-1mm}
\subsubsection{Graph of baseline model}\label{sec:baseline_graph}
The baseline model \eqref{eq:baseline_model} can be represented with a computational graph as %follows. We describe $\zbasek = \mathrm{vec}(\xbasek,\uk)$ and $\theta_\text{base}$ as input variables, and $\wbasek = \mathrm{vec}(\xbasekplus,\yestk)$ as output variables. Then the computational problem can be described as
\vspace*{-15pt}
\begin{multline}
    \begin{bmatrix}
        w_{b,1} \\
        \vdots \\
        w_{b,n_{x_b}} \\
        w_{b,n_{x_b} + 1} \\
        \vdots \\
        w_{b,n_{x_b} + n_{y}}
    \end{bmatrix} = 
    \begin{bmatrix}
        \phi_{\mathrm{base}_1}(\theta_{\mathrm{base},1},z_{\text{b},k,1}) \\
        \vdots \\
        \phi_{\mathrm{base}_{n_{x_b}}}(\theta_{\mathrm{base},n_{x_b}},z_{\text{b},k,n_{x_b}}) \\
        \phi_{\mathrm{base}_{n_{x_b} + 1}}(\theta_{\mathrm{base},n_{x_b}+1},z_{\text{b},k,n_{x_b}+1}) \\
        \vdots \\
        \phi_{\mathrm{base}_{n_{x_b} + n_{y}}}(\theta_{\mathrm{base},n_{x_b}+n_y},z_{\text{b},k,n_{x_b}}+n_y)
    \end{bmatrix} \\
    = \begin{bmatrix}
        f_\mathrm{base}(\theta_\mathrm{base},\zbasek)\\
        h_\mathrm{base}(\theta_\mathrm{base},\zbasek)
    \end{bmatrix},
\end{multline}
\vskip -5mm
where we describe $\zbasek = \mathrm{vec}(\xbasek,\uk)$ and $\theta_\text{base}$ as input variables, $\wbasek = \mathrm{vec}(\xbasekplus,\yestk)$ as output variables, and
 $\theta_{\mathrm{base},1}, \ldots, \theta_{\mathrm{base},n_{x_b}+n_y} \in \theta_{\mathrm{base}}$ as the parameters (i.e., the parameters may be shared between functions $\phi_{\mathrm{base}_i}$), and similarly $z_{\text{b},k,1}, \ldots, z_{\text{b},k,n_{x_b}+n_y} \in \zbasek$.

\begin{figure}
    \centering
    \includegraphics[width=0.2525\textwidth]{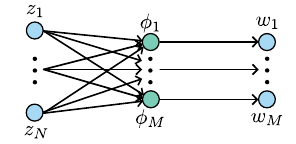}\vspace{-2mm}
    \caption{Computational graph for a computation problem consisting of $M$ functions $\phi_1 \ldots \phi_M$.}
    \label{fig:comp_graph_Phi} \vspace{-1.5mm}
\end{figure}
Which subset of $\theta_{\mathrm{base}}$ and $\zbasek$ are input to each function $\phi_{\mathrm{base}_i}$, which follows from the baseline model \eqref{eq:baseline_model} considered. After deriving a computational graph $G_{\phi_{\mathrm{base}_i}}$ for each function $\phi_{\mathrm{base}_1}, \ldots, \phi_{\mathrm{base}_{n_{x_b} + n_{y}}}$, the computational graph $G_{\phi_\mathrm{base}}$ is derived by the graph union as in \eqref{eq:graph_union}, with the resulting nodes $V_{z_\mathrm{b}}$, $V_{w_\mathrm{b}}$ and $V_{\phi_\text{base}}$. Notably, $E_{\zbase \phi_\text{base}}$ is sparse, e.g., the number of edges between $V_{z_\mathrm{b}}$ and $V_{\phi_\text{base}}$ is smaller than would be permissible by any arbitrary chosen baseline model $\eqref{eq:baseline_model}$.

% \textcolor{red}{
% The adjacency matrix $P_\text{base}$ connecting the input nodes $z_b$ to the functions $\phi_\text{base}$ can alternatively be obtained by
% \begin{equation}
%  P_{\text{base},oi} =
% \begin{cases}
% 1 & \text{if } \frac{\partial \phi_{\text{base},o}}{\partial z_{b,i}} \neq 0\\
% 0 & \text{otherwise}.
% \end{cases}
% \end{equation}
% }

\subsubsection{Graph of the learning component}\vspace{-1mm}
Similarly as for the baseline model, we describe the learning function graph for the computational problem as \vspace{-1.5mm}
\begin{equation*}
    \begin{bmatrix}
        \hspace{-5pt}w_1 \\
        \vdots \\
        w_{n_{w_a}} \hspace*{-3pt}
    \end{bmatrix} \! \hspace*{-2pt} = \hspace*{-2pt} \!  
    \begin{bmatrix}
        \hspace*{-4pt} \phi_{\mathrm{aug}_1}(\theta_{\mathrm{aug},1},z_{\text{a},k,1}) \\
        \vdots \\
        \phi_{\mathrm{aug}_{n_{w_a}}}(\theta_{\mathrm{aug},n_{w_a}},z_{\text{a},k,n_{w_a}}) \hspace*{-2pt}
    \end{bmatrix} \hspace*{-3pt}
    = \hspace*{-1pt} \phi_\text{aug}(\theta_\mathrm{aug}, \zaugk), \vspace{-1.5mm}
\end{equation*}
Then the computational graph $G_{\phi_\mathrm{aug}}$ is derived as before, with the resulting nodes $V_{z_\mathrm{a}}$, $V_{w_\mathrm{a}}$ and $V_{\phi_\text{aug}}$.

% One-to-one functions represent activation functions. The parameters $\theta$ can again be represented by additional input nodes $\zeta$.
% %
% \begin{align}
%     \waugk = \phi_\mathrm{aug}(\theta_\mathrm{aug}, \zaugk) &= \mathcal{N}(\theta_\mathrm{aug}, \zaugk),
% \end{align}
% %
% This leaves the input nodes $z_{a,1}, \ldots, z_{a,n_a}$ and the output nodes $w_{a,1}, \ldots, w_{a,n_a}$.

\subsubsection{LFR model graph}\vspace{-1mm}
With the baseline model graph $G_{\phi_\text{base}}$ and the learning component graph $G_{\phi_\text{aug}}$, we can define the interconnection structure between these graphs and the signals $\hat x_k$, $\uk$, $\hat x_{k+1}$ and $\yestk$ of the model. For ease of notation, we will leave out the $\theta$ nodes, as these are not altered any further. We start by defining how the model signals are matched to $G_{\phi_\text{base}}$ and $G_{\phi_\text{aug}}$, by taking the following parameterised summation \vspace{-3mm}
\begin{equation}\label{eq:sum_node_equation}
    \varsigma = \sum_{i=1}^{\indeg{\varsigma}}\theta_i v_i, \vspace{-2mm}
\end{equation}
where $\varsigma \in \text{disjoint}\left(V_{x^+},V_{y},V_{\zbase},V_{\zaug}\right)$ and $v_i \in \text{disjoint}\left(V_{x},V_{u},\right.$ $\left.V_{\wbase},V_{\waug}\right)$ and $\theta_i$ is the weight of the summation. These summation operations are then represented by the summation nodes $V_\varsigma$. Then the interconnect graph can be defined as

\begin{defn}[Interconnection graph]\label{definition:interconnection_graph}
The computational graph of the interconnection between $G_{\phi_\text{base}}$ and $G_{\phi_\text{aug}}$ is a directed graph denoted by $G_\text{LFR} = \left(V_\text{LFR},E_\text{LFR}\right)$, where $V_\text{LFR}= \text{disjoint}\left(V_{x^+},V_{y},V_{\zbase},V_{\zaug},V_{x},V_{u},V_{\wbase},V_{\waug}\right)$ and $E_\text{LFR}= \text{disjoint}\left(E_{\varsigma_b \zbase}, E_{\varsigma_a \zaug}, E_{\varsigma_x x^+}, E_{\varsigma_y y}, E_{\phi_\text{base}}, E_{\phi_\text{aug}}\right)$
\begin{enumerate}[label=(\alph*)]
    \item\label{th:a} $V_{\varsigma_b} = \{\varsigma_1, \ldots, \varsigma_{n_{\zbase}}\}$, $V_{\varsigma_a} = \{\varsigma_1, \ldots, \varsigma_{n_{\zaug}}\}$, \\ $V_{\varsigma_x} = \{\varsigma_1, \ldots, \varsigma_{n_{x}}\}$, $V_{\varsigma_y} = \{\varsigma_1, \ldots, \varsigma_{n_{y}}\}$
    \item\label{th:b} $E_{\varsigma_b \zbase} = \mathrm{outvar}\left(V_{\varsigma_b}, V_{\zbase}\right)$, $E_{\varsigma_a \zaug} = \mathrm{outvar}\left(V_{\varsigma_a}, V_{\zaug}\right)$ \\
    $E_{\varsigma_x x^+} = \mathrm{outvar}\left(V_{\varsigma_x}, V_{x^+}\right)$, $E_{\varsigma_y y} = \mathrm{outvar}\left(V_{\varsigma_y}, V_{y}\right)$
    % \item\label{th:b} $E_{\varsigma_b \zbase}\!=\!\left\{\left(\varsigma_b, \zbase\right)| \varsigma_b \in V_{\varsigma_b}, \zbase \in V_{\zbase}, \outdeg{\varsigma_b} = \indeg{ \zbase} = 1 \right\}$
    % \item\label{th:c} $E_{\varsigma_a \zaug} = \left\{\left(\varsigma_a, \zaug\right)| \varsigma_a \in V_{\varsigma_a}, \zaug \in V_{\zaug}, \outdeg{\varsigma_a} = \indeg{ \zaug} = 1 \right\}$
    % \item\label{th:d} $E_{\varsigma_x x^+} = \mathrm{outvar}\left(V_{\varsigma_x}, V_{x^+}\right)$, $E_{\varsigma_y y} = \mathrm{outvar}\left(V_{\varsigma_y}, V_{y}\right)$
    % \item\label{th:d} $E_{\varsigma_x x^+} = \left\{\left(\varsigma_x, x^+\right)| \varsigma_x \in V_{\varsigma_x}, x^+ \in V_{x^+}, \outdeg{\varsigma_x} = \indeg{x^+} = 1 \right\}$
    % \item\label{th:e} $E_{\varsigma_y y} = \left\{\left(\varsigma_y, y\right)| \varsigma_y \in V_{\varsigma_y}, y \in V_{y}, \outdeg{\varsigma_y} = \indeg{y} = 1 \right\}$
    % \item\label{th:source_LFR} $\forall v \in V_x \cup V_u, \indeg{v} = 0$
    % \item\label{th:sink_LFR} $\forall v \in V_{x^+} \cup V_y, \outdeg{v} = 0$
    \item\label{th:source_LFR} $\forall v \in V_x \cup V_u, \indeg{v} = 0$, $\forall v \in V_{x^+} \cup V_y, \outdeg{v} = 0$
    % \item\label{th:source_LFR} $V_{xu} = \mathrm{source}\left(V_x \cup V_u\right)$, $V_{x^+y} = \mathrm{sink}\left(V_{x^+}  \cup V_y\right)$, 
    % \item $\forall v \in \text{disjoint}\left(V_{w_b}, V_{w_a}, V_{x_{k+1}}, V_y\right), \enspace \exists! \varsigma \in V_\varsigma, \left(\varsigma v\right) \in E_{\varsigma, v}$
    \item\label{th:connections_LFR} $E_{v \varsigma} = \left\{\left(v,\varsigma\right)| v \in \text{disjoint}\left(V_{z_b}, V_{z_a}, V_{x_k}, V_{u_k}\right), \varsigma\in V_\varsigma\right\}$
\end{enumerate}
\end{defn}
\ref{th:a} defines the summation nodes for the baseline, augmentation, state, and output nodes and \ref{th:b} defines the one-to-one outgoing edges from these summation nodes to the respective output nodes $\rho$. \ref{th:source_LFR} defines existing nodes as inputs and outputs respectively (i.e., sources and sinks). Finally, \ref{th:connections_LFR} defines the edges between the state, input, baseline, and augmentation and the summation nodes.

\begin{figure}
    \centering
    \includegraphics[width=0.45\textwidth]{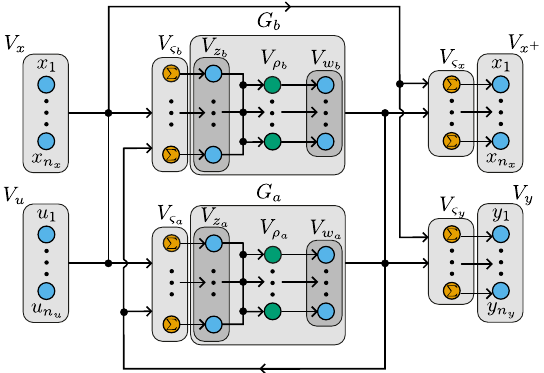}
    \caption{Computational graph of the interconnection of $G_b$ and $G_a$.}
    \label{fig:comp_graph_Phi} \vspace{-1mm}
\end{figure}
By applying vertex contraction with the nodes $\varsigma_i \in V_{\varsigma_i}$ and the nodes to which the corresponding outgoing edge of $\varsigma$ leads, either $V_z$, $V_{x^+}$ or $V_{y}$, on on $G_\text{LFR}$, and also applying vertex contraction with node sets $V_{\rho_b}$ and $V_{\rho_a}$ on $G_\text{LFR}$, we retrieve a graph with only the variable nodes. Considering $P_i$ to be the adjacency matrix for the edge set $E_i$, then the adjacency matrix for the interconnected graph is  \vspace{-4mm}
\begin{equation} \label{eq:adjacency_matrix}
    P = \begin{bNiceMatrix}[first-row,first-col] & \hat{x}_k & u_k & \zbase & \zaug & \wbase & \waug & \hat{x}_{k+1} & \hat{y}_k \\ \hat{x}_k & 0 & 0 & 0 & 0 & 0 & 0 & 0 & 0 \\ u_k & 0 & 0 & 0 & 0 & 0 & 0 & 0 & 0 \\ \zbase & \mathit{P_{\zbase x}} & \mathit{P_{\zbase u}} & 0 & 0 & \mathit{P_{\zbase \wbase}} & \mathit{P_{\zbase \waug}} & 0 & 0 \\ \zaug & \mathit{P_{\zaug x}} & \mathit{P_{\zaug u}} & 0 & 0 & \mathit{P_{\zaug \wbase}} & \mathit{P_{\zaug \waug}} & 0 & 0 \\ \wbase & 0 & 0 & P_\mathrm{b} & 0 & 0 & 0 & 0 & 0 \\ \waug & 0 & 0 & 0 & P_\mathrm{a} & 0 & 0 & 0 & 0 \\ \hat{x}_{k+1} & \mathit{P_{x x}} & \mathit{P_{x u}} & 0 & 0 & \mathit{P_{x \wbase}} & \mathit{P_{x \waug}} & 0 & 0 \\ \hat{y}_k & \mathit{P_{y x}} & \mathit{P_{y u}} & 0 & 0 & \mathit{P_{y \wbase}} & \mathit{P_{y \waug}} & 0 & 0 \\ \end{bNiceMatrix}.  \vspace{-2mm}
\end{equation}
By considering all edge sets maximal, i.e., fully connected where allowed, the summations in \eqref{eq:sum_node_equation} recover the proposed LFR structure \eqref{eqs:general_LFR}. The augmentation structures in Tables~\ref{tab:augmentation_structures} and \ref{tab:output_augmentation_structures} correspond to sparse adjacency matrices. Thus, we can enforce or detect these augmentation structures by these sparse patterns in the adjacency matrix.

\subsection{Parametrisation of the learning component} \vspace{-1mm}
The proposed LFR-based augmentation structure allows for the use of any parameterised learning function $\phi_\text{aug}$ without loss of generality. For the remainder of this report, however, we will consider $\phi_\text{aug}$ to be parameterised by an ANN. An ANN with \( q \) hidden layers, each composed of \( m_i \) neurons and activation function \( \rho : \mathbb{R} \rightarrow \mathbb{R} \), is defined as \vspace{-1mm}
% \footnote{Consider that each hidden layer is composed of $m$ activation functions $\rho:\mathbb{R} \rightarrow \mathbb{R}$ in the form of $\xi_{i,j} = \rho(\sum_{l=1}^{m_{i-1}}\theta_{\mathrm{w},i,j,l} \xi_{i-1,l}+ \theta_{\mathrm{b},i,j})$ where $\xi_i=\mathrm{col}(\xi_{i,1},\ldots,\xi_{i,{m_i}})$  is the latent variable representing the output of layer $1\leq i\leq q$. Here, $\mathrm{col}(\centerdot)$ denotes the composition of a column vector. For a $\phi_\mathrm{aug}$ with $q$ hidden layers and linear input and output layers, this means that $\phi_\mathrm{aug}(\theta_\mathrm{aug},\zaugk)= \theta_{\mathrm{w},q+1} \xi_q(k) + \theta_{\mathrm{b},q+1}$ and $\xi_{0}(k)=\zaugk$.}
\begin{equation}
    \xi_{i,j} = \rho\left( \sum_{l=1}^{m_{i-1}} \theta_{\mathrm{w},i,j,l} \, \xi_{i-1,l} + \theta_{\mathrm{b},i,j} \right)   \vspace{-2mm}
\end{equation}
where $\xi_i=\mathrm{col}(\xi_{i,1},\ldots,\xi_{i,{m_i}})$  is the latent variable representing the output of layer $1\leq i\leq q$, \( \theta_{\mathrm{w},i,j,l} \) and \( \theta_{\mathrm{b},i,j} \) are the weight and bias parameters of the network. Here, $\mathrm{col}(\centerdot)$ denotes the composition of a column vector.
For a $\phi_\mathrm{aug}$ with $q$ hidden layers and linear input and output layers, this gives   \vspace{-1.5mm}
\begin{subequations}
    \begin{align}
        \phi_\mathrm{aug}(\theta_\mathrm{aug},\zaugk) & = \theta_{\mathrm{w},q+1} \xi_q(k) + \theta_{\mathrm{b},q+1} \\
        \xi_{0}(k) & =\zaugk
    \end{align}    
\end{subequations}
\vskip -4mm
This can be extended to residual neural network (ResNet) \cite{he2016deep}, which adds a linear bypass given as \vspace{-1.5mm}
\begin{subequations} \label{eq:resnet}
    \begin{align}
        \phi_\mathrm{aug}(\theta_\mathrm{aug},\zaugk) & = \theta_{\mathrm{w},q+1} \xi_q(k) + \theta_{\mathrm{b},q+1} + W_a \zaugk \\
        \xi_{0}(k) & =\zaugk
    \end{align}    
\end{subequations}
\vskip -4mm
where $W$ is a parameterised residual weight matrix of appropriate dimensions. The additional linear bypass can capture unknown linear dynamics in the learning functions. This is preferable over capturing these dynamics through nonlinear activation functions, since this is more likely to extrapolate better and has shown better learning performance in similar settings\cite{Schoukens2021}. It additionally allows for stable initialisations of the proposed LFR structure as discussed in Section~\ref{sec:parameter_init}.

This finalises the structure of the fully parameterised general augmentation structure of \eqref{eq:augmentation_structure}. We continue with providing well-posedness conditions and an identification algorithm.

% \begin{rem}
%     As the LFR-based augmentation structure \eqref{eqs:general_LFR} is capable of capturing the unknown linear dynamics in the LFR matrix $\mathbf{W}$, ResNET augmentations are not required to capture the unknown linear dynamics.
% \end{rem}

\vspace{-1mm}
\section{LFR structure well-posedness} \label{sec:well_posedness}\vspace{-1mm}
We now propose conditions under which Property~\ref{condition:well-posedness} is guaranteed.
To derive conditions, we rewrite \eqref{eq:LFR_zb_za} into \vspace{-1.5mm}
\begin{equation}\label{eq:start_elimination}
    \underbrace{\left[\begin{array}{c}
        \zbasek \\
        \zaugk \\
        \end{array}\right] - \mathit{D}_{zw}
    \phi\left(\zbasek, \zaugk \right)}_{v\left(\zbasek,
        \zaugk\right)} = 
        \left[\begin{array}{cc}
        C_z^\mathrm{b} & D_{zu}^\mathrm{b}\\
        C_z^\mathrm{a} & D_{zu}^\mathrm{a}
    \end{array}\right]
    \left[\begin{array}{c}
        \hat x_k \\
        u_k \\
    \end{array}\right], \vspace{-2.5mm}
\end{equation}
where $\text{col}\left(\phi_\text{base}\left(\theta_\text{base}, \zbasek\right),\phi_\text{aug}\left(\theta_\text{aug}, \zaugk\right)\right) = \phi\left(\zbasek, \zaugk \right)$. Then, if the inverse $v^{-1}\left(\zbasek, \zaugk\right)$ exists, \vspace{-2mm}
\begin{equation}
    \left[\begin{array}{c}
        \zbasek \\
        \zaugk \\
    \end{array}\right] = v^{-1}\left(\left[\begin{array}{cc}
        C_z^\mathrm{b} & D_{zu}^\mathrm{b}\\
        C_z^\mathrm{a} & D_{zu}^\mathrm{a}
    \end{array}\right]
    \hspace{-2pt}\left[\begin{array}{c}
        \hat x_k \\
        u_k \\
    \end{array}\right]\right). \vspace{-2.5mm}
\end{equation}
Substitution into \eqref{eqs:state_output_LFR} then eliminates \eqref{eq:LFR_zb} and \eqref{eq:LFR_za}.
% , giving
%
% \begin{multline}\label{eq:elimination}
%     \left[\begin{array}{c}
%     \hat x_{k+1} \\
%     \hat y_k \\
%     \end{array}\right] =
%     \left[\begin{array}{cc}
%     A & B_u\\
%     C_y & D_{yu}
%     \end{array}\right]
%     \left[\begin{array}{c}
%         \hat x_k \\
%         u_k \\
%     \end{array}\right] + \\
%     \left[\begin{array}{cc}
%     B_w^\mathrm{b} & B_w^\mathrm{a}\\
%     D_{yw}^\mathrm{b} & D_{yw}^\mathrm{a}
%     \end{array}\right]
%     \phi \circ v^{-1}\left(\left[\begin{array}{cc}
%         C_z^\mathrm{b} & D_{zu}^\mathrm{b}\\
%         C_z^\mathrm{a} & D_{zu}^\mathrm{a}
%     \end{array}\right]
%     \hspace{-2pt}\left[\begin{array}{c}
%         \hat x_k \\
%         u_k \\
%     \end{array}\right]\right),
% \end{multline}
% where $\circ$ stands for function composition.
The existence of the inverse $v^{-1}\left(\zbasek, \zaugk\right)$ can be guaranteed by the following theorem. For this, we introduce the notation $C^n$ to denote the class of functions whose derivatives up to order $n$ exist and are continuous.
\begin{thm}[Hadamard's global inverse function \cite{krantz2002implicit}]\label{thm:Hadamard}
    Let the function $v(z): \mathbb{R}^N \to \mathbb{R}^N$ be a $C^2$ mapping. Suppose that the determinant of the Jacobian $\det\left(D v(z)\right) \neq 0, \forall z\in\mathbb{R}^N$. In addition, suppose that $v(z)$ is proper, i.e., $\|v(z)\|^2_2 \to \infty$ as $\|z\|^2_2 \to \infty$. Then, there exists an inverse function \( v^{-1}(z) \) such that \( v(v^{-1}(z)) = z \) and \( v^{-1}(v(z)) = z \).

\end{thm}
The conditions of this theorem can be met through a variety of parameterisations of the interconnect matrix and the functions $\phi_i(z_i)$ \cite{winston2020monotone, revay2020contracting}. We prove that the conditions of the theorem can be met by the following conditions on the computational graph $G_\text{LFR}$ and the functions $\phi_\text{base}$ and $\phi_\text{aug}$. 
\begin{condition}[Directed acyclic graph]\label{condition:acyclic_LFR}
The computational graph $G_\text{LFR}$ is acyclic, i.e., it contains no cycles.
\end{condition}
%
% \begin{condition}[Bounded function]\label{condition:bounded}
% The function $\phi: \mathbb{R}^n \to \mathbb{R}^m$ is bounded if $\exists M \in \mathbb{R}^+$ such that $|\phi(z)| \leq M, \enspace \forall z \in \mathbb{R}^n $
% \end{condition}
%
\begin{rem}
    For a given adjacency matrix ${P}$, the presence of a topological ordering, and thus the acyclic property, can be computed in linear time $\mathcal{O}(n)$\cite{bang2008digraphs}.
\end{rem}
\begin{condition}[Differentiability]\label{condition:differentiable}
The functions $\phi_\text{base}$ and $\phi_\text{aug}$ are $C^2$.
\end{condition}
Under these conditions, the following theorem holds:
\begin{thm}[Well-posedness of the augmentation]
    Given a parameterisation $\theta = \left(\theta_\text{LFR}, \theta_\text{base}, \theta_\text{aug}\right)$ of the model augmentation structure \eqref{eqs:general_LFR}, the parameterised structured is well-posed if Conditions~\ref{condition:acyclic_LFR} and \ref{condition:differentiable} hold. \vspace{-6mm}
    % under $\phi_\text{base}$ and $\phi_\text{aug}$, if these functions are continuously differentiable.
    % \begin{equation}
    %     \|\left[\begin{array}{c}
    %     \phi_\text{base}\left(\theta_\text{base}, \zbasek\right) \\
    %     \phi_\text{aug}\left(\theta_\text{aug}, \zaugk\right) \\
    % \end{array}\right]\| 
    % \|v\left(\zbasek, \zaugk\right)\|
    % \to \infty \text{ as }
    % \|\begin{bmatrix}
    %     \zbasek \\ \zaugk
    % \end{bmatrix}\| \to \infty
    % \end{equation}
    % and 
    % \begin{equation}
    %      \phi_\text{base} \text{ and } \phi_\text{aug} \text{ are continuously differentiable.}
    % \end{equation}
\end{thm}
\begin{pf}
By Condition~\ref{condition:differentiable}, $\phi_\text{base}$ and $\phi_\text{aug}$ are $C^2$, and also   \vspace{-2.0mm}
\begin{equation}
    v\left(\zbasek, \zaugk\right) = \begin{bmatrix}
        \zbasek \\ \zaugk
    \end{bmatrix} + D_{zw}\left[\begin{array}{c}
        \phi_\text{base}\left(\theta_\text{base}, \zbasek\right) \\
        \phi_\text{aug}\left(\theta_\text{aug}, \zaugk\right) \\
    \end{array}\right]  \vspace{-2.0mm}
\end{equation}
is $C^2$. This satisfies the first condition of Hadamard's global inverse function theorem.

Second, the properness of $v\left(\zbasek, \zaugk\right)$. By Condition~\ref{condition:acyclic_LFR}, we can, by elementary row and column operations, retrieve a strict lower triangular of \( D_{zw} D \phi(z) \). This implies that the function \(\bar \phi(z) =  D_{zw}\phi(z) \) has without loss of generality, the structure \vspace{-1.5mm}
\[
\bar \phi_i(z) = \bar \phi_i(z_{1}, \dots, z_{i-1}), \quad i=1,\ldots,n_z, \vspace{-1.5mm}
\]
with $\bar \phi(z) = \mathrm{col}\left(\bar \phi_1, \ldots, \bar \phi_{n_z}\right)$. Each component of \( v \) is given by \vspace{-1.5mm}
\begin{equation}\label{eq:proof_v_triang}
    v_i(z) = z_i + \bar \phi_i(z_1, \dots, z_{i-1}), \quad i=1,\ldots,n_z. \vspace{-0.5mm}
\end{equation}
We then prove that $\|v(z)\|^2_2 \to \infty$ as $\|z\|^2_2 \to \infty$. From $\|z\|^2_2~\to~\infty$, we have that at least one $|z_i|\to \infty$ while the remainder may remain bounded. We thus prove that this $z_i$ induces $\|v(z)\|^2_2 \to \infty$. For $i=1$, \( \bar \phi_1(z) \) is constant, %(it does not depend on $z$), 
therefore, $v_1(z) = z_{1} + c_1$ for some constant \( c_1 \in \mathbb{R} \). Then $|v_1(z)| \to \infty$ as $|z_1| \to \infty$. For $i= 2,\ldots,n_z$, if $\bar \phi_i$ is bounded, then by \eqref{eq:proof_v_triang} $|v_i(z)| \to \infty$ as $|z_i| \to \infty$. Thus if any $|z_i| \to \infty$, and thus $\|z\|^2_2 \to \infty$, then $\|v(z)\|^2_2$.
%
% We then prove by inductive contradiction that $\|v(z)\|^2_2 \to \infty$ as $\|z\|^2_2 \to \infty$. Suppose that a contradiction exists where there is \( z \in \mathbb{R}^{n_z} \) such that
% \[
% \|z\|^2_2 \to \infty \quad \text{but} \quad \|v(z)\|^2_2 \leq M
% \]
% for some constant \( M > 0 \). This implies that for each \( i \)
% \[
% |v_i(z)| = |z_{i} + \bar \phi_i(z_{1}, \dots, z_{i-1})| \leq M.
% \]
% For $i=1$, \( \bar \phi_1(z) \) is constant (it does not depend on $z$), therefore,
% \[
% v_1(z) = z_{1} + c_1
% \]
% for some constant \( c_1 \in \mathbb{R} \). Then \( |z_{1}| = |v_1(z) - c_1| \leq M + |c_1| \), so \( z_{1} \) is bounded.

% Next, the induction step. Suppose \( z_{1}, \dots, z_{i-1} \) are bounded. Since \( \bar \phi_i(z_{1}, \dots, z_{i-1}) \) depends only on these bounded variables, it is itself bounded. Since \( v_i(z) = z_{i} + \bar \phi_i(z) \) is also bounded, it follows that \( z_{i} \) is bounded.

% By induction, if $\|v(z)\|^2_2$ is bounded, then also \( z \) is bounded, contradicting the assumption that \( \|z\|^2_2 \to \infty \). Therefore, for any \( \{z\}\in \mathbb{R}^{n_z} \) with \( \|z\|^2_2 \to \infty \), we must have \( \|v(z)\|^2_2 \to \infty \).

Third, we prove the Jacobian determinant condition. The Jacobian determinant of $v\left(\zbasek, \zaugk\right)$ is written as \vspace{-1.5mm}
\begin{equation}
    \det\left(D v\left(\zbasek,
        \zaugk\right)\right) = \det\left(\mathrm{I} - \mathit{D_{zw}} D \phi\left(\zbasek,
        \zaugk\right)\right). \vspace{-1.5mm}
\end{equation}
By Condition~\ref{condition:acyclic_LFR}, we know that all subgraphs of $G_\text{LFR}$ are acyclic. Thus the feedback connection, given as $\mathit{D_{zw}} \text{diag}(P_b, P_a)$,
% \begin{equation}
%     \mathit{D_{zw}}\begin{bmatrix}
%     P_b & 0 \\
%     0 & P_a
%     \end{bmatrix}
% \end{equation}
must be acyclic and thus nilpotent. Then, by the construction of $P_b$ and $P_a$, also the matrix $\mathit{D_{zw}} \text{diag}(D \phi_{\text{base}}, D \phi_{\text{aug}})$
% \begin{equation}\label{eq:acyclic_jacobian}
%     \mathit{D_{zw}}\begin{bmatrix}
%     D \phi_{\text{base}} & 0 \\
%     0 & D \phi_{\text{aug}}  
% \end{bmatrix}
% \end{equation}
is nilpotent, implying  %By the properties of nilpotent matrices 
\vspace{-1.5mm}
\begin{equation}
    \det\left(\mathrm{I} - \mathit{D_{zw}} \begin{bmatrix}
    D \phi_{\text{base}} & 0 \\
    0 & D \phi_{\text{aug}}  
    \end{bmatrix}\right) \neq 0. \vspace{-3.5mm}
\end{equation}
Then $v\left(\zbasek,\zaugk\right)$ has an inverse $v^{-1}$ and the algebraic loop can be eliminated, proving the well-posedness of \eqref{eqs:general_LFR}. %the model augmentation structure~\eqref{eqs:general_LFR}.  
\hfill $\blacksquare$ \vspace{-2mm}
\end{pf}
\begin{rem}\label{rem:LFR_WP_options}
    Condition~\ref{condition:acyclic_LFR}, which ensures the well-posedness of the LFR-based structure in \eqref{eqs:general_LFR}, can be satisfied through several strategies. The simplest approach is to restrict $D_{zw}\equiv~0$ during model training. However, this constraint limits the generality of the method by reducing the variety of model augmentation structures that the LFR-based representation can capture. A more flexible alternative is to include only one of the components, either $D_{zw}^\mathrm{ba}$ or $D_{zw}^\mathrm{ab}$, in the parameter vector $\theta_\mathrm{LFR}$, while setting all other elements of $D_{zw}$ to zero. This approach preserves the acyclic condition while allowing for a richer class of model augmentation structures to be represented by the LFR-based formulation. The selection between tuning $D_{zw}^\mathrm{ba}$ or $D_{zw}^\mathrm{ab}$ is a modelling choice and should be guided by prior physical insights. It is left to future work to develop parameterisations that directly guarantee well-posedness of the model structure.
\end{rem}

\vspace{-1mm}\section{Identification Algorithm} \label{sec:Identification}\vspace{-1mm}
Given the proposed model augmentation structure \eqref{eqs:general_LFR}, our goal is to estimate the parameters of the model structure based on measured data in order to capture the behaviour of the data-generating system \eqref{eq:nl_dyn}. To this end, we specify an identification algorithm consisting of identification criterion, baseline parameter regularisation, data and baseline model normalisation, and parameter initialisation. 
% Measurement data and model testing will be handled in the examples in Sections~\ref{sec:SimulationStudy} and \ref{sec:RealWorldIdent} as these depend on the application considered.

% \begin{figure}
%     \centering
%     \includegraphics[width=1.0\linewidth]{Figures/MA_identification_cycle.pdf}
%     %\vspace{2pt}
%     \caption{Identification cycle with proposed LFR-based model augmentation structure (grey) and components described in this chapter (dashed grey).}
%     \label{fig:ma_identification_cycle}
% \end{figure}

% \subsection{Identification criterion}
% The identification criterion consists of two components. First, we introduce the truncated loss function. This is a simulation based loss with beneficial additional properties. Second, we introduce a regularisation term for the baseline model parameters.
\vspace{-1mm}
\subsection{Truncated loss function}\label{sec:truncated_loss}\vspace{-2mm}
We adapt the multiple-shooting-based truncated objective function~\cite{gerben2022}. This is a truncated prediction loss based objective function. The data~$\mathcal{D}_N$ is split into $N$~subsections of length~$T$. This allows for the use of computationally efficient batch optimisation methods popular in machine learning, while also increasing data efficiency~\cite{gerben2022}. This truncated objective function is given as \vspace{-2mm}
\begin{subequations} \label{eq:loss_function}
    \hspace{-5mm}  \begin{align}
         \hspace*{-5mm}  V_\text{trunc}(\theta) \hspace*{-0mm} \! = \! \hspace*{-0mm}  &\frac{1}{N-T+1} \hspace*{-2mm} \sum_{i=1}^{N-T+1} \hspace*{-1.mm} \frac{1}{T}\sum_{\ell=0}^{T-1}\left\|\hat{y}_{k_i+\ell \vert k_i}\!-\!y_{k_i+\ell}\right\|_2^2 \\
    \begin{bmatrix}
        \hat{x}_{k_i+\ell + 1 \vert k_i} \\
        \hat{y}_{k_i+\ell \vert k_i} \\
        w_{\mathrm{b},k_i+\ell \vert k_i} \\
        w_{\mathrm{a},k_i+\ell \vert k_i}
    \end{bmatrix} &:= W(\theta_\mathrm{LFR}) \begin{bmatrix}
        \hat{x}_{k_i+\ell \vert k_i} \\
        u_{k_i+\ell} \\
        z_{\mathrm{b},k_i+\ell \vert k_i} \\
        z_{\mathrm{a},k_i+\ell \vert k_i}
    \end{bmatrix} \\
    w_{\mathrm{b},k_i+\ell \vert k_i} & := \phi_\text{base}\left(\theta_\text{base}, \zbasek\right)\\
    w_{\mathrm{a},k_i+\ell \vert k_i} & := \phi_\text{aug}\left(\theta_\text{aug}, \zaugk\right)\\
    \hat{x}_{k \vert k} & :=\mathnormal{\psi}\left(\theta_\text{encoder}, y_{k-n_a}^{k-1}, u_{k-n_b}^{k-1}\right),
    \end{align}
\end{subequations}
\vskip -4mm
where $\theta = \text{col}(\theta_\text{base}, \theta_\text{aug}, \theta_\mathrm{LFR}, \theta_\text{encoder})$ is the joint parameter vector, $k+\ell \vert k$ indicates the state $\hat{x}_k$ or the output $\hat{y}_k$ at time $k + \ell$ simulated from the initial state $\hat{x}_{k \vert k}$ at time $k$. The subsections start at a randomly selected time $k_i\in \{n+1,\ldots,N-T\}$. The initial state of these subsections is estimated by an encoder function $\psi$ using past input-output data, i.e., $\hat{x}_{k|k}=\psi(\theta_\text{encoder}, y_{k-n_a}^{k-1}, u_{k-n_b}^{k-1})$ where $u_{k-n_b}^{k-1}= [\ u_{k-n_b}^\top \ \cdots \ u_{k-1}^\top ]^\top$ for $\tau\geq 0$ and $y_{k-n_a}^{k-1}$ is defined similarly. The model structure with the encoder $\psi$ included, is shown in Fig.~\ref{fig:n-step-encoder-graphic}. 

The existence of the encoder $\psi$ has been shown in \cite{beintema2024data} for state-space models. We give a brief overview  of the underlying mechanism. We derive the $n$-step ahead predictor of the data generating system \eqref{eq:nl_dyn}, resulting in \vspace{-2.5mm}
%
%\begin{subequations}
    \begin{equation}
        y_k^{n+k}
        = \underbrace{\begin{bmatrix}
            h(x_k, u_k) + e_k \\
            (h \circ f)(x_k, u_k^{k+1}) + e_{k+1} \\
            \vdots \\
            (h \circ_n f)(x_k, u_k^{n+k}) + e_{k+n} \\
        \end{bmatrix}}_{\Gamma_n\left(x_k, u_k^{n+k}\right) + e_k^{n+k}}.  \vspace{-2.5mm}
    \end{equation}
%\end{subequations}
%
where $\circ_n$ stands for $n$ times recursive function composition. If $\Gamma_n$ is partially invertible as $x_k = \mathnormal{\Phi}_n(u_k^{n+k}, y_k^{n+k} - e_k^{n+k})$, then the reconstructability map \cite{isidori1985nonlinear} is given as  \vspace{-2.5mm}
\begin{subequations}
    \begin{align}
        x_k & = (\circ_n f)(x_{k-n}, u_{k-n}^{k}) \\
        & = (\circ_n f)(\mathnormal{\Phi}_n(u_{k-n}^{k}, y_{k-n}^{k}-  e_{k-n}^{k}), u_{k-n}^{k}) \\
        & = \mathnormal{\Psi}_n(u_{k-n}^{k}, y_{k-n}^{k} -  e_{k-n}^{k}). \label{eq:encoder_derived}
    \end{align}
\end{subequations}
\vskip -4.5mm
However, the noise sequence $e_{k-n}^{k}$ is not available in practice. Under the assumption that  $e_k$ is i.i.d. white noise, we can use the conditional expectation of \eqref{eq:encoder_derived} as an estimate of $x_k$:  \vspace{-2mm}
\begin{equation}
    \bar x_k = \mathbb{E}_{e_k}\left[x_k | u_{k-n}^{k}, y_{k-n}^{k}\right] = \bar{\mathnormal{\Psi}}(u_{k-n}^{k}, y_{k-n}^{k}),  \vspace{-1.5mm}
\end{equation}
which is an unbiased estimator of $x_k$\cite{janny2022learning}. This estimator is difficult to compute in practice due to the required analytical inversion to obtain $\mathnormal{\Phi}_n$, which varies with the choices of $\theta_\text{base}$, $\theta_\text{aug}$ and $\theta_\mathrm{LFR}$. Instead the parameterised function estimator $\psi$ is used to co-estimate $\bar{\mathnormal{\Psi}}$ with $W(\theta_\mathrm{LFR})$, $\phi_\text{base}$ and $\phi_\text{aug}$.
\begin{figure}
    \centering
    \includegraphics[width=1.0\linewidth]{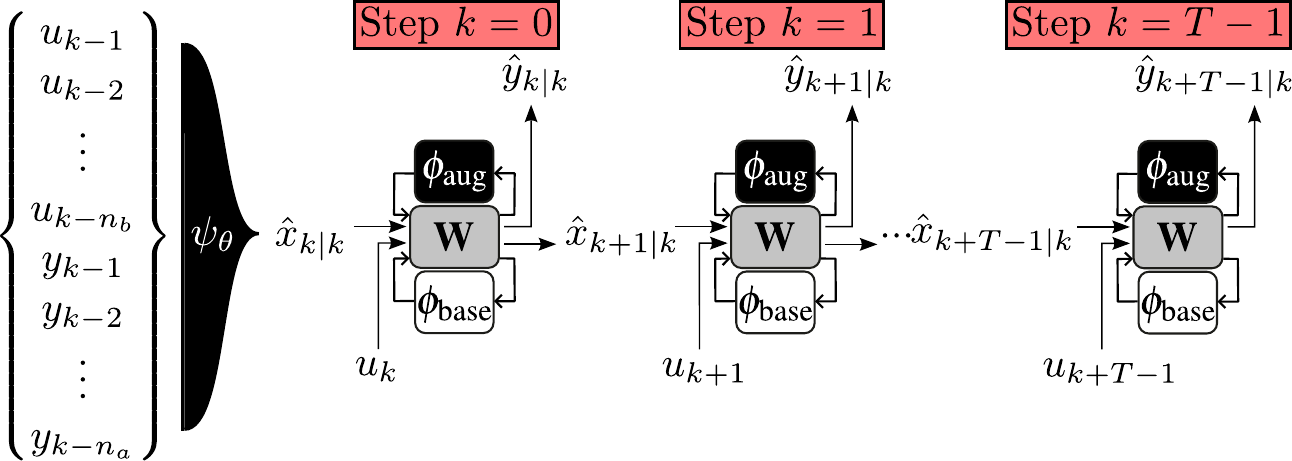}
    %\vspace{2pt}
    \caption{SUBNET structure: the subspace encoder $\psi_{\theta}$ estimates the initial state at time $k$ based on past inputs and outputs, then it is propagated through $\phi_\theta$ multiple times until a simulation length $T$.}
    \label{fig:n-step-encoder-graphic} \vspace{-1.5mm}
\end{figure}

 \vspace{-1mm}
\subsection{Baseline parameter regularisation} \vspace{-1.5mm}
The proposed LFR-based model structure \eqref{eqs:general_LFR} is overparameterised and, as a result, the optimal parameter values $\left(\theta_\mathrm{base}^\ast,\theta_\mathrm{aug}^\ast,\theta_\mathrm{LFR}^\ast\right)$ that minimise \eqref{eq:loss_function} are not unique. Therefore, the joint identification of all parameters can result in the learning components representing or cancelling out part of the baseline model dynamics.
These non-unique parameterisations have similar model performance; however, some parameterisations may generate physical parameters that deviate wildly from the expected parameters, even leading to physically unrealistic parameter values. To address this issue and retain the interpretability of the baseline model, we adapt the regularisation cost term from \cite{bolderman_feedforward_2022, bolderman_physics-guided_2024}:  \vspace{-2.5mm} %, which is given as
\begin{equation}\label{eq:regularization}
    V^\text{reg}(\theta) = \left\|\Lambda \left(\theta_\text{base} - \theta_\text{base}^0\right)\right\|_2^2  \vspace{-2.5mm}
\end{equation}
where $\Lambda=\lambda \mathrm{diag}(\theta_\mathrm{base}^0)^{-1}$ with $\lambda \in \mathbb{R}_{\geq 0}$ as a tunable parameter. This regularisation term \eqref{eq:regularization} penalises deviations of the baseline parameters from the a priori selected values $\theta_\mathrm{base}^0$, with the diagonal element normalising the importance of the parameters compared to each other. The cost function with the regularisation term becomes  \vspace{-2.5mm}
\begin{equation}\label{eq:regularized_cost}
    V_{\mathcal{D}_N}(\theta) = V_{\mathcal{D}_N}^\text{trunc}(\theta) + V^\text{reg}(\theta).  \vspace{-2.5mm}
\end{equation}
The tunable variable $\lambda$ determines how much the baseline parameters can deviate from the nominal parameter set $\theta_\mathrm{base}^0$, relative to the change in $T$-step-ahead prediction cost \eqref{eq:loss_function}.
%
% This results in a model augmentation structure where the learning component augments the first-principle model. At the same time, deviations of the physical parameters are allowed when it results in better data fit. 
% Furthermore, the Lambda matrix $\Lambda$ specifies the relative importance of each baseline model parameter. It also serves as a normalisation layer when it is defined as $\Lambda=\lambda \mathrm{diag}(\theta_\mathrm{base}^0)^{-1}$, where $\lambda \in \mathbb{R}_{\geq 0}$ is a tunable parameter. The value of $\lambda$ determines how much the baseline parameters can deviate from the nominal parameter set $\theta_\mathrm{base}^0$, relative to the change in $T$-step-ahead prediction cost \eqref{eq:loss_function}.
% It can be also interpreted as a trade-off between model performance and parameter interpretability. If $\lambda=0$, the original objection function \eqref{eq:loss_function} is retained, while for $\lambda \rightarrow\infty$ only \eqref{eq:regularization} remains active with a trivial solution of $\theta_\mathrm{base}^0$.
%
% \begin{rem}
%     As an alternative approach, an orthogonal projection-based regularisation is proposed for model augmentation in \cite{gyorok_orthogonal_2025} to avoid competition between the physics-based and learning components. However, that method is currently limited to the static parallel model augmentation structure.
% \end{rem}
 \vspace{-1mm}
\subsection{Data and baseline model normalisation} \vspace{-1.5mm}
For the estimation of model structures containing ANNs, normalisation of the input and output data to zero mean and to a standard deviation of~1 has been shown to improve model estimation \cite{bishop1995neural}. Therefore, we normalise $u$ and $y$ in $\mathcal{D}_N$ and aim to initialise the model structure and \eqref{eqs:general_LFR} so that $\hat x$ is also normalised as in \cite{gerben2022}. For this, the to-be-augmented baseline model $\phi_\text{base}$ needs to be considered in the normalisation process.
%, see \cite{dehkordi2026normalization} for an example on parallel augmentation.
We take $\phi_\text{base}$ into consideration based on the work in \cite{Schoukens2020}. This results in a model that takes normalised input and state, and returns a normalised output while not altering its dynamics. For a baseline model operating around zero mean $\hat x$, $u$, $\hat y$, this can be achieved by the transformation  \vspace{-2.5mm}
\begin{subequations}
    \begin{align}
    \bar{f}_{\text{base}} & = T_x f_\text{base}\left(\theta_\text{base}, T_x ^{-1} \xbase, T_u ^{-1} u \right) \\
    \bar{h}_{\text{base}} & = T_y h_\text{base}\left(\theta_\text{base}, T_x ^{-1} \xbase, T_u ^{-1} u \right),
\end{align}
\end{subequations}
\vskip -4.5mm
where $T_u \in \mathbb{R}^{n_u}$ is a diagonal matrix composed as $T_\mathrm{u}^{-1}= \text{diag}(\sigma_{u_1}^{-1},\ldots, \sigma_{u_{n_\mathrm{u}}}^{-1})$, where $\sigma_u$ is the sample-based standard deviation of each input signal, computed based on the data set $\mathcal{D}_N$. The transformation matrix $T_y$ is defined in a similar way. For $T_x$, the standard deviation $\sigma_x$ is determined on a sequence of baseline states $\xbase$. For this, we simulate the baseline model with $\theta_\mathrm{base}^0$ for the given input sequence resulting in $\hat{x}_{b,k}$. We define an extended data set $\hat{\mathcal{D}}_N = \{ \left({y}_k, \hat{x}_{\mathrm{b},k}, u_k\right)\}_{k=1}^N$ for initialisation of the encoder later.
% For this, we complete dataset $\mathcal{D}_N$ by including the simulation of the baseline model for the given input sequence, i.e., $\hat{\mathcal{D}}_N = \{ \left({y}_k, \hat{x}_{\mathrm{b},k}, u_k\right)\}_{k=1}^N$.

 \vspace{-1mm}
\subsection{Model structure initialisation}\label{sec:parameter_init}  \vspace{-1.5mm}
Next, we consider the initialisation of the learning components, LFR matrix, and the encoder. A common approach in the literature for initialising ANNs is to randomly assign weights and biases. However, for baseline models, random initialisation can be unstable and result in poor optimisation results. Furthermore, initialisation based on prior information, e.g., baseline model behaviour, can improve convergence rate and enhance model accuracy.
% However, for model augmentation structures, initialising based on the baseline model and conserving the baseline model behaviour at start of training, has been shown to improve convergence time and enhance model accuracy. Furthermore, random initialisation can be unstable and thus result in poor optimisation. To preserve the baseline model behaviour, we initialise the encoder with the baseline model and consider pass-through initialisations for the learning components

 \vspace{-1mm}
\subsubsection{Model behaviour at initialisation} \vspace{-1mm}
We propose to initialise the parameters $\theta$ so that the LFR-based model structure with the encoder behaves equivalent to the baseline model on initialisation. We note the model structure \eqref{eqs:general_LFR} with $\zbasek$ and $\zaugk$ eliminated as $\Omega$. Then the initialisation of $\theta$ should realise the following behaviour  \vspace{-2.5mm}
%
% \begin{multline}\label{eq:baseline_behaviour}
%     \left[\begin{array}{c}
%     \hat x_{k+1} \\
%     \hat y_k \\
%     \end{array}\right] =
%     \left[\begin{array}{cc}
%     A & B_u\\
%     C_y & D_{yu}
%     \end{array}\right]
%     \left[\begin{array}{c}
%         \hat x_k \\
%         u_k \\
%     \end{array}\right] + \\
%     \left[\begin{array}{cc}
%     B_w^\mathrm{b} & B_w^\mathrm{a}\\
%     D_{yw}^\mathrm{b} & D_{yw}^\mathrm{a}
%     \end{array}\right]
%     \phi \circ v^{-1}\left(\left[\begin{array}{cc}
%         C_z^\mathrm{b} & D_{zu}^\mathrm{b}\\
%         C_z^\mathrm{a} & D_{zu}^\mathrm{a}
%     \end{array}\right]
%     \hspace{-2pt}\left[\begin{array}{c}
%         \hat x_k \\
%         u_k \\
%     \end{array}\right]\right) \\
%     = \begin{bmatrix}
%         f_\mathrm{base}(\theta_\mathrm{base},\zbasek)\\
%         \tilde{\mathcal{N}}(\theta_\mathrm{aug},\zaugk) \\
%         h_\mathrm{base}(\theta_\mathrm{base},\zbasek)
%     \end{bmatrix},
% \end{multline}
\begin{equation}\label{eq:baseline_behaviour}
    \left[\begin{array}{c}
    x_{\text{b},k+1} \\
    x_{\text{a},k+1} \\
    \hat y_k \\
    \end{array}\right] = \Omega \left(\theta, \zbasek, \zaugk\right) = \begin{bmatrix}
        f_\mathrm{base}(\theta_\mathrm{base}^0,\zbasek)\\
        \phi_\text{aug}(\theta_\mathrm{aug}^0,\zaugk) \\
        h_\mathrm{base}(\theta_\mathrm{base}^0,\zbasek)
    \end{bmatrix},  \vspace{-2.5mm}
\end{equation}
where $\theta_\mathrm{aug}^0$ are the initialised parameters for the augmentation. It is not trivial to determine such an initialisation for any arbitrary baseline model, learning function, and LFR matrix combination. We propose here a method that can achieve this initialisation under the following conditions:
\begin{enumerate}[label=(\alph*)]
    \item \label{th:acyclic}The computational graph $G_\text{LFR}$ is acyclic
    \item \label{th:linear}The learning functions are ResNets \eqref{eq:resnet}.
\end{enumerate}
Condition \ref{th:acyclic} is the same condition as for the well-posedness proof and thus is not restrictive. Condition \ref{th:linear} is required to feasibly create series augmentations that can have baseline model behaviour.
 \vspace{-1mm}
\subsubsection{Encoder initialisation} \vspace{-1.5mm}
% We assume that the encoder network is parameterised by a ResNET, then the encoder-based state estimation can be expressed as
% \begin{multline}
%     \begin{bmatrix}
%         x_{\mathrm{b},k\vert k}\\x_{\mathrm{a},k\vert k}
%     \end{bmatrix} = \psi_\theta\left(y_{k-n_a}^{k-1}, u_{k-n_b}^{k-1}\right) = \\\tilde{\psi}_{\tilde{\theta}}\left(.,\, .\right) + \begin{bmatrix}
%         W_{\psi,y}^\mathrm{b} & W_{\psi,u}^\mathrm{b}\\
%         W_{\psi,y}^\mathrm{a} & W_{\psi,u}^\mathrm{a}
%     \end{bmatrix} \begin{bmatrix}
%         y_{k-n_a}^{k-1}\\ u_{k-n_b}^{k-1}
%     \end{bmatrix},
% \end{multline}
% where $\tilde{\psi}$ is a feedforward ANN with $\tilde{\theta}$ parameters, where $W_{\psi,y}^\mathrm{b}$, and $W_{\psi,u}^\mathrm{b}$ represent the linear parts of the network corresponding to the baseline state reconstruction, while $W_{\psi,y}^\mathrm{a}$, and $W_{\psi,u}^\mathrm{a}$ are utilised for reconstructing the augmented states.

The encoder is parameterised by an ANN, e.g., a ResNET. To guarantee baseline behaviour at initialisation, this encoder should at initialisation output the baseline state sequence as in the extended dataset $\hat{\mathcal{D}}_N$. This could be derived analytically as in Section~\ref{sec:truncated_loss}, but this is complicated on not feasible for all ANNs. Instead, we consider a data-driven approach. We fit a baseline encoder $\psi_\text{base}$ on this dataset using the following loss function during the initialisation step  \vspace{-2.5mm}
\begin{equation}\label{eq:initialisation_loss}
    V_\text{enc}(\theta) = \frac{1}{N} \sum_{k=1}^N \left\| \psi_\text{base}\left(\theta_\text{base},\hat y_{k-n_a}^{k-1}, u_{k-n_b}^{k-1}\right) -\hat{x}_{\mathrm{b},k}\right\|_2^2,  \vspace{-2.5mm}
\end{equation}
where $\hat{x}_{\mathrm{b},k}$ is the forward simulated state of the baseline model in the extended dataset $\hat{\mathcal{D}}_N$. If augmented states are considered, the baseline encoder $\psi_\text{base}$ is extended with an augmented state encoder $\psi_\text{aug}$  \vspace{-2.5mm}
\begin{equation}
    \begin{bmatrix}
        {x}_{\text{b},k \vert k} \\
        {x}_{\text{a},k \vert k}
    \end{bmatrix} = \begin{bmatrix}
        \psi_\text{b}\left(\theta_\text{base},\hat y_{k-n_a}^{k-1}, u_{k-n_b}^{k-1}\right) \\
        \psi_\text{a}\left(\theta_\text{aug},\hat y_{k-n_a}^{k-1}, u_{k-n_b}^{k-1}\right)
    \end{bmatrix},  \vspace{-2.5mm}
\end{equation}
where the weights and biases of $\psi_\text{aug}$ are initialised by the Xavier approach. The loss function \eqref{eq:initialisation_loss} is no longer considered after initialisation.
 \vspace{-1mm}
\subsubsection{LFR matrix and learning component initialisation} \vspace{-1.5mm}
% \begin{subequations}\label{eqs:general_LFR}
% \begin{align}
%     \left[\begin{array}{c}
%         \hat{x}_{k+1}\\ \hat{y}_k \\ \hdashline[5pt/2pt]  \zbasek\\ \zaugk
%     \end{array}\right] &= \underbrace{\left[\begin{array}{cc;{5pt/2pt}cc}
%         A & B_u & B_w^\mathrm{b} & B_w^\mathrm{a}\\
%         C_y & D_{yu} & D_{yw}^\mathrm{b} & D_{yw}^\mathrm{a}\\ \hdashline[5pt/2pt] 
%         C_z^\mathrm{b} & D_{zu}^\mathrm{b} & D_{zw}^\mathrm{bb} & D_{zw}^\mathrm{ba}\\
%         C_z^\mathrm{a} & D_{zu}^\mathrm{a} & D_{zw}^\mathrm{ab} & D_{zw}^\mathrm{aa}
%     \end{array}\right]}_{W(\theta_\mathrm{LFR})} \begin{bmatrix}
%         \hat{x}_k\\ u_k\\ \hdashline[5pt/2pt]  \wbasek\\ \waugk
%     \end{bmatrix},\\
%         \wbasek & = \phi_\text{base}\left(\theta_\text{base}, \zbasek\right),\\
%         \waugk & = \phi_\text{aug}\left(\theta_\text{aug}, \zaugk\right),
% \end{align}
% \end{subequations}
We now initialise the LFR matrix and the learning components so that \eqref{eq:baseline_behaviour} holds for initialisation. We can simplify this equation with Condition \ref{th:linear} by only considering the linear component of the learning function and initialising the NL component to be zero, i.e., $\phi_\text{aug}\left(\zaugk\right) = 0 + W_a \zaugk$. We further assume under the acyclic property, without loss of generality, that $D_{zw}$ is lower block diagonal with respect to the learning functions and the baseline model. %\footnote{This reasoning can be extended to divisions of the baseline model and learning functions, resulting in the same trade-off with more matrices.}. 
This means that $D_{zw}^\mathrm{aa} = 0$, $D_{zw}^\mathrm{bb} = 0$, and either $D_{zw}^\mathrm{ba}=0$ or ${D}_{zw}^\mathrm{ab}=0$. Substituting the linear component of the learning function into \eqref{eq:start_elimination} and eliminating $\zaugk$, gives  \vspace{-2.5mm}
\begin{equation}
    \zbasek =  \underbrace{\left[\begin{array}{cc}
    C_z^\mathrm{b} + D_{zw}^\mathrm{ba} W_a{C}_z^\mathrm{a} & {D}_{zu}^\mathrm{b} + D_{zw}^\mathrm{ba} W_a {D}_{zu}^\mathrm{a}
    \end{array}\right]}_{\tilde{C}} \left[\begin{array}{c}
        \hat x_k \\
        u_k \\
    \end{array}\right],  \vspace{-2.5mm}
\end{equation}
and the prediction equation  \vspace{-2.5mm}
\begin{multline}
    \left[\begin{array}{c}
    \hat x_{k+1} \\
    \hat y_k \\
    \end{array}\right] =
    \underbrace{\left[\begin{array}{cc}
    {A} + B_w^\mathrm{a} W_a {C}_z^\mathrm{a} & {B_u} + B_w^\mathrm{a} W_a {D}_{zu}^\mathrm{a}\\
    {C_y} + D_{yw}^\mathrm{a} W_a {C}_z^\mathrm{a} & {D_{yu}} + D_{yw}^\mathrm{a} W_a {D}_{zu}^\mathrm{a}
    \end{array}\right]}_{\tilde{A}}
    \left[\begin{array}{c}
        \hat x_k \\
        u_k \\
    \end{array}\right] + \\
    \underbrace{\left[\begin{array}{cc}
    B_w^\mathrm{b} + B_w^\mathrm{a} W_a {D}_{zw}^\mathrm{ab}\\
    D_{yw}^\mathrm{b} +  D_{yw}^\mathrm{a} W_a  {D}_{zw}^\mathrm{ab} 
    \end{array}\right]}_{\tilde{B}}
    \phi_\text{base}\left(\zbasek\right). 
\end{multline}
\vskip -4mm
Thus, to have an initialisation satisfying \eqref{eq:baseline_behaviour}, we require $\tilde{B} = I_{n_x+n_y}$ and $\tilde{C} = I_{n_x+n_u}$. We repeat here that we assume either $D_{zw}^\mathrm{ba}=0$  or ${D}_{zw}^\mathrm{ab}=0$ to satisfy Condition~\ref{condition:acyclic_LFR}. The choice between these assumptions results in initialisation similar to the model structures derived in Appendix~\ref{appendix:general_LFR_augmentation_forms}, with $D_{zw}^\mathrm{ba}=0$ resulting in series output augmentations, ${D}_{zw}^\mathrm{ab}=0$ in series-input and $D_{zw}^\mathrm{ba}={D}_{zw}^\mathrm{ab}=0$ in parallel augmentation.

All matrices not required to set the baseline model behaviour at initialisation \eqref{eq:baseline_behaviour} have all elements $m$ of the matrix initialised randomly, according to \cite{Schoukens2020}, i.e., $m \sim \mathcal{U}(-1,1)$ where $\mathcal{U}(a,b)$ denotes a uniform distribution with support from $a$ to $b$.  \vspace{-2.5mm}
\subsection{Convergence and Consistency} \vspace{-1.5mm}
% The notions of \emph{convergence} and \emph{consistency} follow \cite{ljung1978convergence} and their specialisation in \cite{gerben2022}.
Next, we can analyse the statistical properties of the introduced augmentation approach in terms of convergence and consistency \cite{ljung1978convergence}. \emph{Convergence} implies that, as the number of samples in $\mathcal{D}_N$ tends to infinity, the empirical identification criterion approaches its expected value. An estimator is \emph{consistent} if, as $N \to \infty$, the estimated model converges to an equivalent representation of the true system \eqref{eq:nl_dyn}.

Under Conditions~2.1--2.4 in \cite{gerben2022} on the data-generating process \eqref{eq:nl_dyn}, model structure \eqref{eqs:general_LFR}, and identification criterion \eqref{eq:loss_function}, the convergence and consistency are proven. Condition~2.1 requires the data-generating system to be incrementally exponential output stable, which becomes an assumption on the considered system \eqref{eq:nl_dyn}. Similarly, Conditions 2.3 on predictor convergence and 2.4 on persistence of excitation are commonly assumed to hold. However, Condition~2.2 warrants additional consideration, as it demands differentiability of the model structure, including the encoder in \eqref{eq:loss_function}, with respect to $\theta$. This is not trivial due to the feedback connection in the model structure. If we take the acyclic condition (Condition~\ref{th:acyclic}), then differentiability of $v(\zbasek,\zaugk)$ is ensured, and it remains only to assume differentiability of the prediction mappings \eqref{eqs:state_output_LFR}, $\phi_\text{base}$, and $\phi_\text{aug}$—a standard and non-restrictive assumption in system identification. Thus, all conditions in \cite{gerben2022} are fulfilled to imply convergence and consistency of the proposed model augmentation approach.

 \vspace{-1mm}\section{Simulation Study} \label{sec:SimulationStudy} \vspace{-1mm}
% \OL{General specification of what is done here.
%
% In this section we analyse the ability of the different model augmentation realisations of the LFR on different systems with a baseline model. By doing this we show that the LFR is able to estimate the systems starting from the baseline model, while having a faster estimation and also being more transparent than the fully ANN-SS model. We consider here structured parameterisations of the LFR, e.g.,  the ones shown in Table~\ref{tab:augmentation_structures}and Table~\ref{tab:output_augmentation_structures}. This is then further expanded to more free parameterisations in the next Section with a real world setup and a more complex baseline mode, where it is not a priori clear what would be the best form of augmentation.}
%
In this section, we analyse the performance of different model augmentation realisations of the LFR model when applied to various systems using a common baseline model. The objective is to demonstrate that the proposed LFR augmentation structure together with the proposed identification methods can effectively estimate system behaviour starting from the baseline model, while achieving faster estimation and offering greater transparency compared to the fully ANN-SS model. 
We focus here on structured parameterisations of the LFR in the form of those presented in Table~\ref{tab:augmentation_structures} and Table~\ref{tab:output_augmentation_structures}. In Section~\ref{sec:RealWorldIdent}, this analysis is extended to flexible parameterisations, applied to a real-world setup with a more complex baseline model, where the most suitable form of augmentation is not known a priori.
 \vspace{-1mm}
\subsection{Mass-Spring-Damper system and data generation} \vspace{-1.5mm}
% As a simulation example, a  MSD system composed from 3 masses under various configurations is considered as shown in Fig. \ref{fig:msd_systems} with the physical parameters given in Table \ref{tab:params_msd}. All systems have a hardening spring nonlinearity, the second also has input saturation, and the third one has a first-order low-pass filter applied on the output. The states associated with the system representation are the positions $p_i$ and the accelerations $\dot p_i$ of the masses $m_1$, $m_2$, $m_3$. The hardening spring is described as a cubic nonlinearity with parameter $a_1$. The position $p_2$ is the measured output and the input force is applied on the first mass.
%
As a simulation example, a \emph{mass–spring–damper} (MSD) system consisting of three masses is considered under three different configurations, as illustrated in Fig.~\ref{fig:msd_systems}. The corresponding physical parameters are listed in Table~\ref{tab:params_msd}. Configuration (a) consists of three masses and a hardening spring nonlinearity. Configurations (b) and (c) add to this an input saturation and a first-order low-pass filter (LPF) to the output, respectively.

Configuration (a) is described in terms of a state-space representation with a total of 6 states as the positions $p_i$ and velocities $\dot{p}_i$ of the masses $m_1$, $m_2$, and $m_3$. The hardening spring nonlinearity is simulated as a cubic stiffness term multiplied by the parameter $a_1$. The measured output corresponds to the position $p_2$, and the external input force is applied to the first mass $m_1$. For Configuration (b) the saturation is $30\tanh(\frac{u}{30})$ which results in a reduction of the applied multisine input signal RMS from $10.0$ N  to $9.11$ N. For Configuration (c), the LPF is simulated such that the Bode plot of the LPF component is as shown in Fig.~\ref{fig:bodeplot_lpf}.

% These three systems demonstrate the variety of different model augmentation structures that can be required and that the proposed augmentation structure can represent.

\begin{figure*}
    \begin{subfigure}[]{0.475\textwidth}
        \centering 
        \includegraphics[scale=0.5]{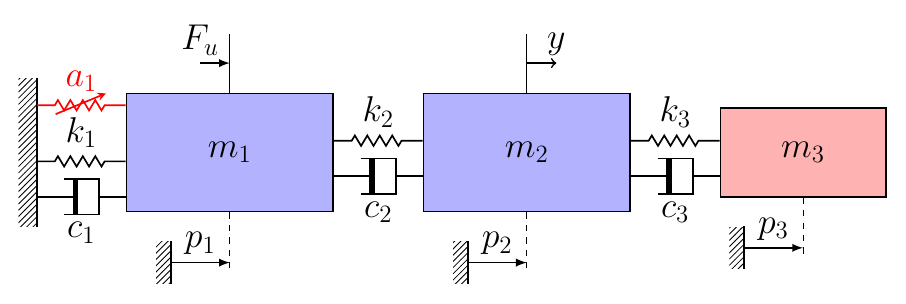}
        \caption{3-DOF MSD with nonlinear cubic spring}    
        \label{fig:msd_nonlinear}
    \end{subfigure}
    \hfill
    \begin{subfigure}[]{0.475\textwidth}  
        \centering 
        \includegraphics[scale=0.5]{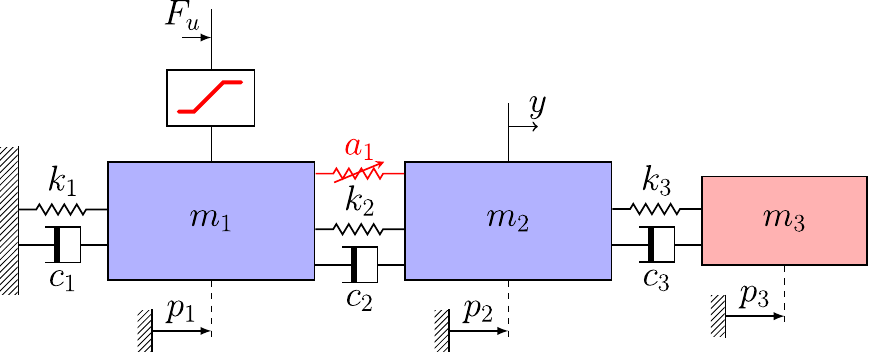}
        \caption{3-DOF MSD with nonlinear cubic spring and input saturation} 
        \label{fig:msd_input}
    \end{subfigure}
    \vskip -1mm
    
    \begin{subfigure}[]{0.475\textwidth}   
        \centering 
        \includegraphics[scale=0.5]{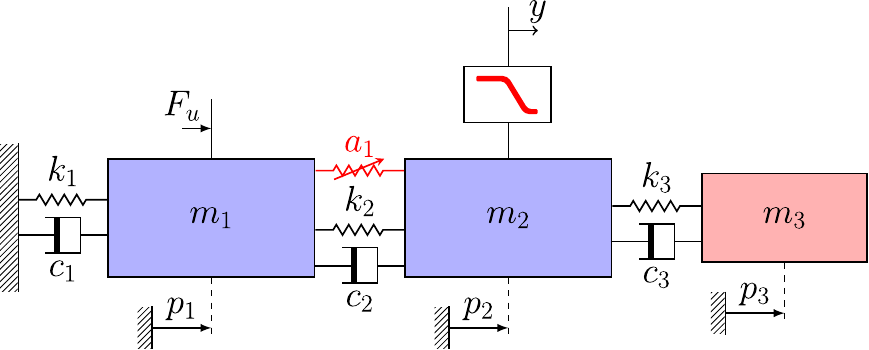}
        \caption{3-DOF MSD with nonlinear cubic spring and low-pass filter}  
        \label{fig:msd_output}
    \end{subfigure}
    \hfill
    \begin{subfigure}[]{0.38\textwidth}   
        \centering
        \includegraphics[scale=0.5]{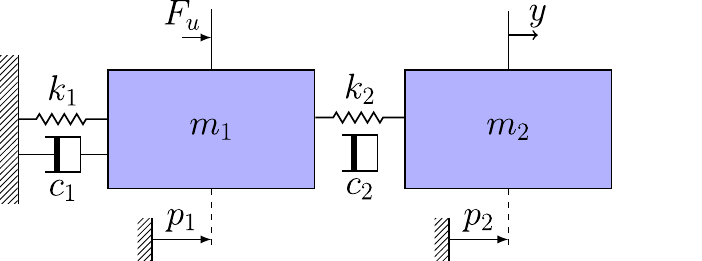}
        \caption{2-DOF MSD linear baseline model}
        \label{fig:msd_baseline}
    \end{subfigure}
    \caption[]
    {Considered baseline model in blue and black and the MSD data generating systems with additional dynamics in red.} 
    \label{fig:msd_systems}
\end{figure*}
    
% \begin{figure}
%     \centering
%     \includegraphics[width=0.49\textwidth]{Figures/3dofMSD_LPF.pdf}
%     \caption{3-DOF MSD. The linear dynamics of the $m_1$ and $m_2$ are assumed known, while the dynamics of $m_3$ and the contribution of $a_1$ are unknown.}
%     \label{fig:3dof_msd}
% \end{figure}

\begin{table}
    \centering
    \caption{Physical parameters of the 3-DOF MSD systems.}
    \vspace*{5pt}
    \label{tab:params_msd}
    \begin{tabular}{c||r|r|r|r} 
        Body & Mass $m_i$ & Spring $k_i$ & Damper $c_i$ & Hardening $a_i$ \\
        \hline
        1  & $0.5 \mathrm{~kg}$ & $100 \mathrm{~\frac{N}{m}}$ & $0.5 \mathrm{~\frac{N s}{m}}$  & $100 \mathrm{~\frac{N}{m^3}}$\\
        2 & $0.4 \mathrm{~kg}$  & $100 \mathrm{~\frac{N}{m}}$ & $0.5 \mathrm{~\frac{N s}{m}}$ & -\\
        3 & $0.1 \mathrm{~kg}$  & $100 \mathrm{~\frac{N}{m}}$ & $0.5 \mathrm{~\frac{N s}{m}}$ & -\\

\end{tabular}
\end{table}

\begin{table}
    \centering
    \caption{Approximate physical parameters 2-DOF MSD baseline model.}
    \vspace*{5pt}
    \label{tab:approx_params_msd}
    \begin{tabular}{c||r|r|r} 
        Body & Mass $m_i$ & Spring $k_i$ & Damper $c_i$ \\
        \hline
        1  & $0.5 \mathrm{~kg}$ & $95 \mathrm{~\frac{N}{m}}$ & $0.45 \mathrm{~\frac{N s}{m}}$ \\
        2 & $0.4 \mathrm{~kg}$  & $95 \mathrm{~\frac{N}{m}}$ & $0.45 \mathrm{~\frac{N s}{m}}$ \\

\end{tabular}
\end{table}

% We obtain the data for model estimation by applying 4th~order Runge-Kutta (RK4) based numerical integration on the 3-DOF MSD system with~$T_\mathrm{s}=0.02$~$\mathrm{s}$.

The system is simulated using a 4th order \emph{Runge-Kutta} (RK4) numerical integration with step size $T_s = 0.02$~s and synchronised zero-order-hold actuation and sampling. The values of the DT input signal $u_k$ are generated by a random-phase multisine with 1666 frequency components in the range $[0, 25]$ Hz with a uniformly distributed phase in $ [0, 2\pi)$. The sampled output measurements $y_k$ are perturbed by an additive white noise process $e_k\sim\mathcal{N}(0,\sigma_\mathrm{e}^2)$. Here, $\sigma_\text{e}$ is chosen so that the signal-to-noise ratio is equal to 30 dB. The generated sampled output $y_k$ for the input~$u_{k}$ is collected in the data set $\mathcal{D}_N$. Separate data sets are created for estimation, validation and testing of different realisations of sizes $N_{\text{est}} = 2\cdot10^4$, $N_{\text{val}} = 10^4$, $N_{\text{test}} = 10^4$, respectively. The estimation data comprise two periods, while the validation and testing data each contain a single period.

 \vspace{-1mm}
\subsection{Baseline model} \vspace{-1.5mm}
The baseline model is chosen to represent the linear 2-DOF MSD dynamics shown in Fig.~\ref{fig:msd_baseline}. We consider two initialisations for the baseline model parameters: the ideal values from Table~\ref{tab:params_msd} and the approximate values from Table~\ref{tab:approx_params_msd}. The \emph{root mean squared error}~(RMSE) of the simulation responses of the baseline model for these initial parameter values is shown in Table~\ref{tab:nrms_all} for configurations (a-c). Both initialisations perform relatively poorly, despite (approximately) representing a large part of the dynamics.

% \begin{table}[t]
%     \centering
%     \caption{Hyperparameters for identifying the LFR-based augmentation and ANN-SS models.}
%     \vspace*{4pt}
%     \begin{tabular}{c|c|c|c|c|c|c}
%     \! hidden layers\!&\!nodes\!&\!$n_u$ $n_y$\!&\!$n_a$ $n_b$\!&\!$T$\!&\!epochs\!&\!batch size\!\\
%     \hline
%     2 & 8 & 1 & 7 & 200 & 3000 & 2000 \\
%     \end{tabular}
%     \label{tab:Hyperparam}
% \end{table}

\begin{table}[t]
    \centering
    \caption{Hyperparameters for identifying the LFR-based augmentation and ANN-SS models.}
    \vspace*{4pt}
    \begin{tabular}{c|c|c|c|c|c|c}
    hidden layers
    & nodes
    & $n_a$ $n_b$
    & $T$
    & epochs
    & batch size\\
    \hline
    2 & 8 & 7 & 200 & 3000 & 2000 \\
    \end{tabular}
    \label{tab:Hyperparam}
    \vspace{-8pt}
\end{table}

\vspace{-1.5mm}
\subsection{Config.~(a): MSD with added cubic spring and mass}\label{sec:MSD_a} \vspace{-1.5mm}

% We estimate parameterisations of the LFR-based augmentation structure corresponding to representations of the state augmentation structures described in Table~\ref{tab:augmentation_structures} 
First, we consider the augmentation of the baseline model in a structured form using the introduced LFR model structure. For this we consider various parameterisations of the LFR matrix $W$ corresponding to the configurations in Table~\ref{tab:augmentation_structures} (not including S-SSI or S-DSI as we know a-priori that these will not model the system given the baseline model). For parallel augmentation structures, the learning components are chosen as feedforward neural networks, and for series augmentation structures, we choose ResNets to have a feasible initialisation (see Section~\ref{sec:parameter_init}). For all learning components, the number of hidden layers and neurons are listed in Table~\ref{tab:Hyperparam}. The activation function is chosen as $\tanh$. For the dynamic augmentation, we add two additional states to the baseline model states for a total of 6 states. This is the minimum number of states required to completely model the 3-DOF MSD system. The baseline model, learning component, and encoder parameters are jointly estimated as described in Section~\ref{sec:Method}. The hyperparameters for these estimations are shown in Table~\ref{tab:Hyperparam}, with 16 nodes and 2 hidden layers used for the encoder. The regularisation tuning parameter $\lambda$ in the joint identification cost function was set at $\lambda=1$ for ideal initialisation. The identification criterion is optimised using Adam. As a comparison to a black-box approach, an ANN-SS model parameterised by ResNets is estimated with the SUBNET method \cite{beintema2024data} and hyperparameters from Table~\ref{tab:Hyperparam}.
% Finally, we examine the effectiveness of linear augmentation by estimating a model in S-DP configuration with parameterised matrices instead of feedforward neural networks as augmentation components.

The simulation RMSE on the test data for these estimated models is shown in Table~\ref{tab:nrms_all}. The dynamic augmentations are able to capture the dynamics accurately, while the static augmentations result in slightly less accurate models.
% The linear dynamic augmentations result in significantly worse models. This indicates that the inclusion of additional states and nonlinear augmentations is beneficial for obtaining accurate augmented models of the MSD system with configuration (a) under the considered baseline model.

In Fig.~\ref{fig:val_loss} we show the validation loss curves of select estimated models. The convergence speed of the ANN-SS and S-DSO models is slower than the S-DP model, while all models achieve similar RMSE scores.
% The ANN-SS models achieve an RMSE score similar to those of the dynamic augmentations, as shown in Table~\ref{tab:nrms_all}. However, the S-DP structure achieves this score in relatively few epochs, as shown in the validation loss curves in Fig.~\ref{fig:val_loss}, while the ANN-SS model requires 3000 epochs to reach an equivalent score.

The estimated physical parameters remain very close to the initialisation values for both the ideal and the approximate case. This is not desired behaviour for the approximate initialisations, indicating that the learning components are learning parts of the system dynamics that could be represented by the baseline model.
% Thus, the ideal initialisations are not negated, but the approximate initialisations are also not corrected. This indicates that the learning components are learning part of the system dynamics that could be represented by the baseline model.

\begin{figure}[]
    \hspace*{-20pt}
    \centering
    \includegraphics[width=\linewidth,trim={0 0 0 12pt},clip]{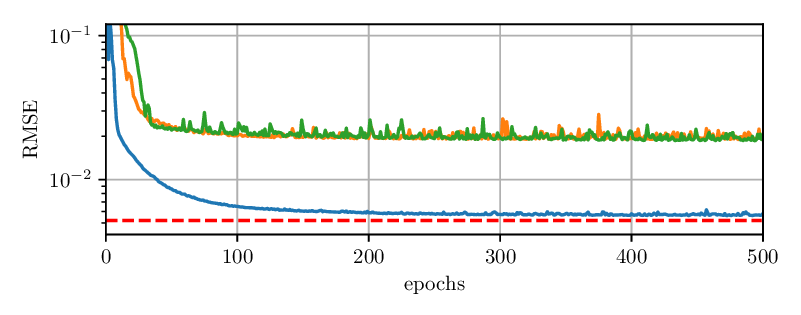}
    \vspace*{-12pt}
    \caption{Validation loss over first 500 training epochs for S-DP (\textcolor{pyplotblue}{\raisebox{0.5mm}{\rule{0.3cm}{0.6mm}}}), S-DSO (\textcolor{pyplotorange}{\raisebox{0.5mm}{\rule{0.3cm}{0.6mm}}}) and black-box  ANN-SS (\textcolor{pyplotgreen}{\raisebox{0.5mm}{\rule{0.3cm}{0.6mm}}}) models estimated for the MSD system with configuration (a). The noise floor is shown with a red dashed line (\textcolor{pyplotred}{\raisebox{0.5mm}{\rule{0.1cm}{0.6mm}} \!\raisebox{0.5mm}{\rule{0.1cm}{0.6mm}}}).}
    \label{fig:val_loss} \vspace{-1.5mm}
\end{figure}

In Fig.~\ref{fig:state_comparison}, we show the comparison between the states~$\hat x$ of the S-DP model (\textcolor{pyplotblue}{\raisebox{0.5mm}{\rule{0.3cm}{0.6mm}}}) and the outputs of the learning components $\phi_\text{aug}$ (\textcolor{pyplotorange}{\raisebox{0.5mm}{\rule{0.3cm}{0.6mm}}}). Here, $x_\text{b}= [x_1^\top \ldots ~ x_4^\top]^\top$ and $x_\text{a} = [x_5^\top x_6^\top]^\top$. The effect of the learning components is relatively small for $x_\text{b}$, while $x_\text{a}$ is modelled solely by the learning components. From this, we can conclude that the learning components are augmenting the baseline and not replacing the baseline model with their own dynamics.
% Therefore, we have obtained an accurate model whose dominant behaviour is given by an interpretable baseline model.
% However, as mentioned above, we cannot state that part of the dynamics that can be represented by the baseline model is learnt instead by the augmentation. To address this issue specifically for the parallel augmentation case, see \cite{gyorok_orthogonal_2025}.

% \begin{figure}[]
%     \centering
%     \includegraphics[width=\linewidth]{Figures/fitted_models_pred_error.eps}
%     \vspace*{-15pt}
%     \caption{Simulation error the nonlinear state augmented model (\textcolor{pyplotorange}{\raisebox{0.5mm}{\rule{0.3cm}{0.6mm}}}) compared to the measured test data (\textcolor{pyplotblue}{\raisebox{0.5mm}{\rule{0.3cm}{0.6mm}}}) for the system (a).}
%     \label{fig:simulation_error}
% \end{figure}

\begin{figure}[]
    \centering
    \includegraphics[width=\linewidth,trim={5pt 10pt 6pt 42pt},clip]{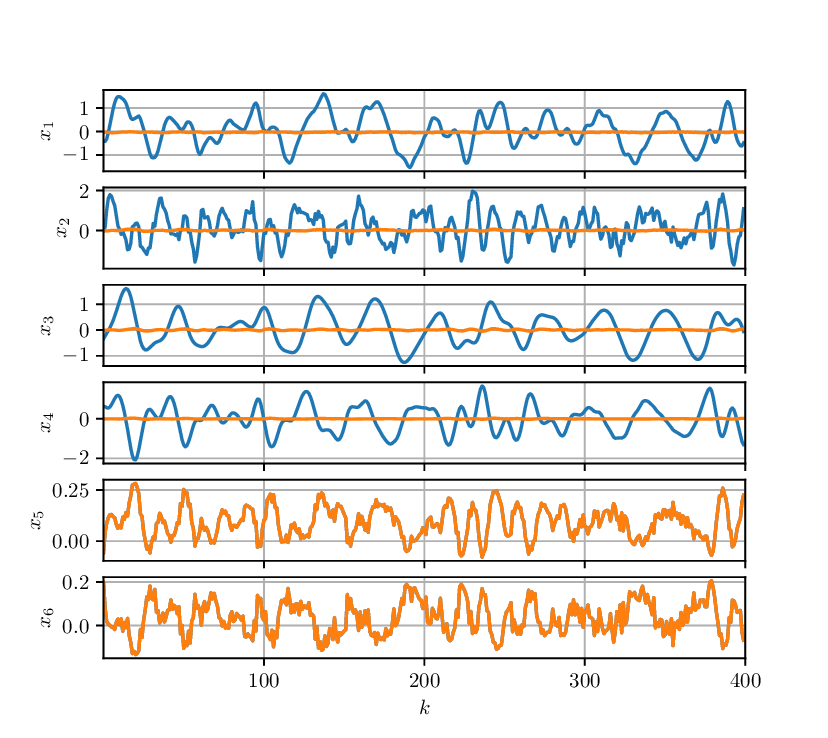}
    \vspace*{-7mm}
    \caption{Comparison of the augmented model with dynamic parallel configuration states $\hat x$ (\textcolor{pyplotblue}{\raisebox{0.5mm}{\rule{0.3cm}{0.6mm}}}) and the outputs of the learning components $\phi_\text{aug}$ (\textcolor{pyplotorange}{\raisebox{0.5mm}{\rule{0.3cm}{0.6mm}}}) for a simulation with test data of the MSD system with configuration (a). States $x_1$ - $x_4$ are  states of the augmented model based on the sum of $f_\text{base}$ and $f_\text{aug}$, while states  $x_5$ - $x_6$ are the output of the dynamic augmentation component $g_\text{aug}$.}
    \label{fig:state_comparison}
\end{figure}

\vspace{-1mm}
\subsection{Config. (b): MSD with added input saturation} \vspace{-1.5mm}
For the MSD system in configuration (b), we applied similar identification steps as in Section~\ref{sec:MSD_a}. We again estimate S-SP and S-DP as described in Table~\ref{tab:augmentation_structures}. We, however, further combine these parameterisations with a series input augmentation to characterise the input saturation, which we note as S-SP-I with the following structure \vspace{-2.0mm}
\begin{equation}
    \xbasekplus =  f_\text{base}\left(\xbasek, g_\text{aug}\left(\uk\right)\right) + f_\text{aug}\left(\xbasek, \uk\right). \vspace{-1.5mm}
\end{equation}
A similar structure is used for a S-DP with series input augmentation noted as S-DP-I. The remaining parameterisations and hyperparameters are as in Section~\ref{sec:MSD_a}. We again estimate an ANN-SS model parameterised by ResNets as a black-box comparison.

The simulation RMSE on the test data for these estimated models are shown in Table~\ref{tab:nrms_all}. All selected augmentations result in accurate models. However, the series input augmentation specifically models the input saturation as a function of $u_k$ and thus results in a more interpretable model. The convergence of the estimated models are similar as in Fig.~\ref{fig:val_loss}, with the S-SP-I and S-DP-I models converging faster than the backbox ANN-SS, S-SP and S-DP models.

\begin{table*}[t]
\centering
\caption{RMSE of the simulated responses from the estimated models evaluated on the test sets generated by the MSD system with configurations (a), (b), and (c).}
\vspace*{-8pt}
\begin{tabular}{c||cc||cc||cc}
\multirow{2}{*}{Model} 
& \multicolumn{2}{c||}{Config. (a)} 
& \multicolumn{2}{c||}{Config. (b)} 
& \multicolumn{2}{c}{Config. (c)} \\
& Ideal & Approx. & Ideal & Approx. & Ideal & Approx. \\
\hline
baseline 
& $1.97 \cdot 10^{-1}$ & $1.77 \cdot 10^{-1}$ 
& $1.96 \cdot 10^{-1}$ & $1.86 \cdot 10^{-1}$ 
& $2.22 \cdot 10^{-1}$ & $2.13 \cdot 10^{-1}$ \\
\hline
S-SP 
& $6.73 \cdot 10^{-3}$ & $6.61 \cdot 10^{-3}$ 
& $5.53 \cdot 10^{-3}$ & $6.03 \cdot 10^{-3}$ 
& $5.86 \cdot 10^{-3}$ & $5.53 \cdot 10^{-3}$ \\
S-DP 
& \bf{$5.54 \cdot 10^{-3}$} & \bf{$5.77 \cdot 10^{-3}$} 
& $5.79 \cdot 10^{-3}$ & $5.80 \cdot 10^{-3}$ 
& $5.46 \cdot 10^{-3}$ & $5.39 \cdot 10^{-3}$ \\
S-SSO 
& $7.07 \cdot 10^{-3}$ & $6.80 \cdot 10^{-3}$ 
& -- & -- 
& -- & -- \\
S-DSO 
& \bf{$5.33 \cdot 10^{-3}$} & \bf{$5.57 \cdot 10^{-3}$} 
& -- & -- 
& -- & -- \\
S-SP-I 
& -- & -- 
& $5.40 \cdot 10^{-3}$ & $5.78 \cdot 10^{-3}$ 
& -- & -- \\
S-DP-I 
& -- & -- 
& $5.44 \cdot 10^{-3}$ & $5.41 \cdot 10^{-3}$ 
& -- & -- \\
S-SP \& O-DSO
& -- & -- 
& -- & -- 
& $5.45 \cdot 10^{-3}$ & $5.42 \cdot 10^{-3}$ \\
S-DP \& O-DSO
& -- & -- 
& -- & -- 
& $5.39 \cdot 10^{-3}$ & $5.38 \cdot 10^{-3}$ \\
\hline
blackbox ANN-SS 
& \multicolumn{2}{c||}{\bf{$5.72 \cdot 10^{-3}$}} 
& \multicolumn{2}{c||}{\bf{$6.37 \cdot 10^{-3}$}} 
& \multicolumn{2}{c}{\bf{$5.55 \cdot 10^{-3}$}} \\
\end{tabular}
\label{tab:nrms_all}
\vspace{-2mm}
\end{table*}

\vspace{-1.5mm}
\subsection{Config. (c): MSD with added output low-pass-filter}\vspace{-1.5mm}
For the MSD system in configuration (c), we applied similar identification steps as in Section~\ref{sec:MSD_a}. We again estimate the S-SP and S-DP described in Table~\ref{tab:augmentation_structures}. We now further estimate these state augmentations with and O-DSO described in Table~\ref{tab:output_augmentation_structures}. The O-DSO is parameterised by an LTI model with a single state, which is capable of modelling the dynamics of the first-order LPF. The remaining parameterisations and hyperparameters as in Section~\ref{sec:MSD_a}. We again estimate an ANN-SS model parameterised by ResNets as a black-box comparison.

The simulation RMSE on the test data for these estimated models is shown in Table~\ref{tab:nrms_all}. All selected augmentations result in accurate models. The convergence of the estimated models are similar as in Fig.~\ref{fig:val_loss}, with the O-DSO models, S-DP and S-SP converging faster than the backbox ANN-SS model.
% We further show in Fig.~\ref{fig:val_loss_output} that the convergence of both augmentations with and without the output augmentation is similar.

In Fig.~\ref{fig:bodeplot_lpf}, the Bode plot of $h_\text{aug}$ from both the  estimated O-DSO model, with S-SP and S-DP state models, are shown compared against the true system LPF. We can see that the output augmentations model behaviour similar to the system LPF thus enhancing the interpretability of the estimated model compared to the model augmentations without the output augmentation where this behaviour will have to be modelled in the state augmentation. Ensuring that output augmentations model the LPF exactly is an identifiability problem left to future research.

\begin{figure}[]
    \vspace{-3.8mm}
    \hspace*{-5pt}
    \centering
    \includegraphics[width=\linewidth,clip]{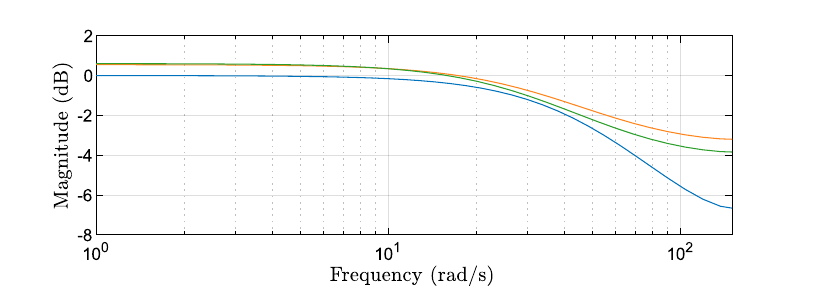}
    \vspace*{-9pt}
    \caption{Bode plot of both O-DSO estimated with the S-DP (\textcolor{pyplotorange}{\raisebox{0.5mm}{\rule{0.3cm}{0.6mm}}}) and S-SP (\textcolor{pyplotgreen}{\raisebox{0.5mm}{\rule{0.3cm}{0.6mm}}}), compared against the LPF included in system configuration (c) (\textcolor{pyplotblue}{\raisebox{0.5mm}{\rule{0.3cm}{0.6mm}}}).}\vspace{-4pt}
    \label{fig:bodeplot_lpf}
\end{figure}

 \vspace{-1mm}\section{Experimental study}\label{sec:RealWorldIdent} \vspace{-1.5mm}
In this section, we demonstrate the capabilities of the general LFR-based model augmentation structure and the proposed estimation approach by identifying the dynamics of an F1Tenth electric vehicle, using experimental data.

\vspace{-1mm}
\subsection{F1Tenth vehicle}
\vspace{-1.5mm}
F1tenth is a 1/10 scale model of an electric car, which has been mainly developed as a test platform for various automotive applications \cite{agnihotri_teaching_2020}. To demonstrate the capabilities of the proposed LFR-based model augmentation structure, the dynamics of such a vehicle are identified in this section. In contrast to Section~\ref{sec:SimulationStudy}, measurements from a real F1tenth are used instead of simulation data. An in-depth description of the used vehicle and test environment is available in \cite{floch_gaussian-process-based_2025}.
%
% The F1tenth car is displayed in Fig.~\ref{fig:f1tenth_meas_setup}. A brushless DC motor drives all four wheels of the vehicle,  and the steering is controlled by a servo motor. A \emph{Vedder Electronic Speed Controller} (VESC) unit is used to drive the actuators. The VESC unit also has a built-in IMU and speed sensor. For measuring the position and the orientation of the vehicle, an OptiTrack motion capture system is used that is capable of high-accuracy measurement of these variables based on passive markers placed on the car. For a more in-depth description of the test environment, refer to \cite{floch_gaussian-process-based_2025}.
%
% \begin{figure}
%     \hspace*{-5pt}
%     \centering
%     \includegraphics[width=0.7\linewidth]{Figures/f1tenth_meas_setup.pdf}
%     \vspace*{-5pt}
%     \caption{F1tenth autonomous small-scale electric vehicle.}
%     \label{fig:f1tenth_meas_setup}
% \end{figure}

\vspace{-1mm}
\subsection{Baseline model of the F1Tenth vehicle}
\vspace{-1.5mm}
To develop a baseline model of the F1Tenth platform, the so-called single-track model has been used \cite{paden_survey_2016}. The model is illustrated in Fig.~\ref{fig:single_track}, and can be expressed by using six state variables. The baseline states are the position of the \emph{center-of-gravity} (CoG) in the $\left(X,Y\right)$ plane $\left(p_x,p_y\right)$, the orientation of the vehicle $\varphi$, which is measured from the $X$ axis, the longitudinal and lateral velocities of the vehicle, and the yaw rate. The control inputs are the steering angle $\delta$, and the PWM percentage of the electric motor that provides the main propulsion of the vehicle. The equations of the single-track model can be derived in continuous time, and the resulting model is discretised using the RK4 scheme. Since the used OptiTrack motion capture system measures the position and orientation of the vehicle, while the built-in IMU and speed sensors provide information regarding the velocity components, full state measurements are available. Hence, the baseline output function becomes $\hat{y}_k = \xbasek$. To model the longitudinal tire force component $F_\xi$, an empirical drivetrain model \cite{floch_gaussian-process-based_2025} is applied, while the linearised Magic Formula \cite{pacejka_chapter_2012} has been utilised to model the lateral tire force components ($F_{\mathrm{f,}\eta}$ and $F_{\mathrm{r,}\eta}$). For a detailed derivation and discussion of the baseline model, refer to \cite{gyorok_orthogonal_2025, floch_model-based_2022}.  %The model is illustrated in Fig.~\ref{fig:single_track}, and expressed as

The applied tire models (especially the empirical drivetrain model) are highly approximative and are the primary sources of inaccuracy in the baseline model; hence, identifying the dynamics of the F1Tenth vehicle is challenging even when incorporating existing physical knowledge into the model structure. Moreover, there are 9 baseline parameters (such as mass, inertia, distance of the rear and front end from the CoG, and parameters corresponding to the tire models) that need to be estimated. Initial values of these parameters were determined in \cite{floch_model-based_2022}. However, some elements of this initial parameter vector $\theta_\mathrm{base}^0$ are highly approximative. Hence, to achieve an accurate representation of the true dynamics, $\theta_\mathrm{base}$ is tuned jointly with the model parameters.

\begin{figure}
    \hspace*{-5pt}
    \centering
    \includegraphics[scale=0.8]{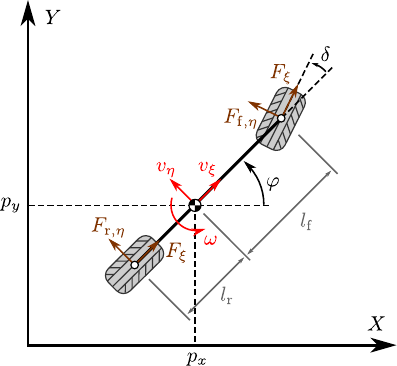}
    \vspace*{-5pt}
    \caption{Illustration of the single track model.}
    \label{fig:single_track} \vspace{-1mm}
\end{figure}

% The baseline states are $x_\mathrm{b}=\begin{bmatrix}
%     p_x & p_y & \varphi & v_\xi & v_\eta & \omega
% \end{bmatrix}^\top$. The control inputs are $\delta$ and $d$. We assume that full-state measurements are available, i.e., the baseline output function can be expressed as
% \begin{equation}
%     \hat{y}_k = \xbasek.
% \end{equation}
% As discussed in Section~\ref{sec:ProblemStatement}, we have assumed that the baseline model dynamics are given in DT form; thus, \eqref{eq:car_EoM} is discretised with the RK4 scheme. 
%The baseline parameters are
% \begin{equation}\label{eq:f1tenth_theta_base}
%     \theta_\mathrm{base}=\begin{bmatrix}
%         m & J_z & l_\mathrm{r} & l_\mathrm{f} & C_\mathrm{m1} & C_\mathrm{m2} & C_\mathrm{m3} & C_\mathrm{r} & C_\mathrm{f}
%     \end{bmatrix}^\top,
% \end{equation}
% and their nominal values have been determined in \cite{floch_model-based_2022}. However, some elements of the determined nominal parameter vector $\theta_\mathrm{base}^0$ are highly approximative; to achieve an accurate representation of the true dynamics, $\theta_\mathrm{base}$ in \eqref{eq:f1tenth_theta_base} is tuned jointly with the ANN parameters.
%

\vspace{-1mm}
\subsection{Data acquisition}
\vspace{-1.5mm}
%The measured outputs have been selected based on the applied baseline model. With the discussed measurement setup, all baseline states and the two input signals are measured directly. 
%$\left(p_x,\,p_y\right)$ position of the CoG in the reference plane $\left(X,\, Y\right)$, $\varphi$ heading angle, that is measured from the $X$ axis, $v_\xi$ longitudinal velocity, $v_\eta$ lateral velocity, and $\omega$ angular velocity (yaw rate).  With the discussed measurement setup, all outputs and the two input signals are measured directly.
%To achieve robust identification results, two different trajectories have been used for training the network.
A lemniscate-shaped trajectory has been selected for generating measurement data because, by following it, the heading angle traverses the whole operational domain, and the resulting motion has quick changes in velocity. A second trajectory has been chosen to be a circle-shaped path because it is also a typical maneuver for this type of vehicle. Measurement data has been collected with a sampling frequency of $f_\mathrm{s}=40$ Hz. To acquire data with various motor PWM inputs, the velocity references have been varied for both trajectories, ranging from 0.45 m/s to 1 m/s with step increments of 0.05 m/s, resulting in a total of 24 measurement records. Half of these measurements have been separated into training and test sets. Both trajectories with alternating reference velocities have been included in all data sets. Before concatenating the measurement signals for the training data set, %20\% of the measured data was randomly selected to form a validation data set.
contiguous segments (corresponding to 20\% of the length of each of the 12 training signals) were randomly selected to form the validation data set. A total of 6467 samples are used for estimation, 1669 for validation, and 8041 for testing.

\subsection{Estimated models}
%As one might see, in the equation of motion \eqref{eq:car_vxi}, \eqref{eq:car_veta}, and \eqref{eq:car_omega} contain the drivetrain and tire dynamics, which hold great amounts of uncertainty, and hence can be the reason for possible model errors. In contrary, \eqref{eq:car_x}, \eqref{eq:car_y}, and \eqref{eq:car_phi} contain simple kinematic relations and can be separated from the rest of the equations.
As proposed by \cite{szecsi_deep_2024, gyorok_orthogonal_2025}, to simplify the neural network structure, we only identify the input-to-velocity relationship by detaching the integrators. Hence, the outputs (and consequently, the baseline states) become the longitudinal and lateral velocity components, as well as the yaw rate of the vehicle. Then, after identification, by putting back the integrator dynamics, the position and orientation values can also be obtained during model simulation.

As all baseline model states are measured, the initial values of $x_\mathrm{b}$ are known for all subsections when calculating the truncated prediction loss, i.e., the encoder network only estimates the augmented states in case a dynamic model augmentation structure is applied. For the latter scenario, a fully-connected feedforward ANN with 2 hidden layers and 64 nodes per layer has been selected for the encoder network with the tanh activation function. Based on previous black-box identification results on the same dataset (see \cite{szecsi_deep_2024}), an encoder lag of $n_a=n_b=12$ has been applied. The augmented state dimension has been selected based on physical insight and a short trial-and-error period as $n_{x_\mathrm{a}}=2$. For the LFR-based structure, $n_{z_\mathrm{a}}=4$ has been applied, while $n_{w_\mathrm{a}}=3$ and $n_{w_\mathrm{a}}=5$ were selected for the static and dynamic model augmentations, respectively. A regularisation coefficient of $\lambda=0.01$ has been applied, as a result of a line-search. All other hyperparameters are summarised in Table~\ref{tab:f1tenth_hyperparams}.

\begin{table}
    \centering
    \caption{Hyperparameters of the LFR-based model augmentation structure for identifying the dynamics of the F1Tenth vehicle.}
    \vspace*{4pt}
    \begin{tabular}{c|c|c|c|c|c}
    \!hidden layers\!&\!nodes\!&\!$n_a$ $n_b$\!&\!$T$\!&\!epochs\!&\!batch size\!\\
    \hline
    2 & 128 & 12 & 40 & 3000 & 256 \\
    \end{tabular}
    \label{tab:f1tenth_hyperparams} 
\vspace{-1.5mm}
\end{table}

As discussed in Remark~\ref{rem:LFR_WP_options}, there are a few possible strategies to ensure the well-posedness of the LFR-based structure \eqref{eqs:general_LFR}. The most straightforward one is to restrict $D_{zw}\equiv 0$, while a more general approach is to select either $D_{zw}^\mathrm{ba}$ or $D_{zw}^\mathrm{ab}$ to be tuned freely, while the rest of $D_{zw}$ is set to zero. To demonstrate these approaches, we have trained models with different options regarding the structure of the $D_{zw}$ matrix.
% Furthermore, to demonstrate the discussed simplification in Section~\ref{sec:user_guidelines}, namely that
Furthermore, we demonstrate the enforcing of model structures in the LFR matrix $W$ by constraining the $C_z^\mathrm{b}$ and $D_{zu}^\mathrm{b}$ matrices such that $\zbasek \equiv \mathrm{vec}\left(\xbasek, \uk\right)$, we have also trained models with that setting. The results are summarised in Table~\ref{tab:f1tenth_results}, where the augmented models are compared with state-of-the-art black-box identification results and the baseline model with nominal parameters. As the integrator dynamics make it difficult to obtain accurate long-term predictions in practice, the presented error values only consider the velocity components of the output. This is in line with previous black-box results using the same data set, see \cite{szecsi_deep_2024}. Still, to demonstrate the accuracy of the simulated position and orientation values, Fig.~\ref{fig:f1tenth_results_comparison} presents these signals obtained using the best-performing static and dynamic LFR-based models for an arbitrarily selected test trajectory. Notably, both models demonstrate remarkable accuracy compared to the real measured data, even over extended open-loop simulations lasting 16–18 seconds. It is important to note that the reported errors are also influenced by the employed numerical integration scheme. This effect is particularly visible in the case of the dynamic LFR model, where the simulated trajectory initially closely follows the measured data, but the accuracy continuously deteriorates due to the accumulation of integration errors and unknown input disturbances.

Further analysing the results shown in Table~\ref{tab:f1tenth_results}, it is visible that the most general options (not fixing $\zbase$, and tuning $D_{zw}^\mathrm{ab}$) have resulted in the best model accuracy for the static structure. 
%However, as the baseline model misses certain states associated with the data-generating system, all dynamic augmentation structures have outperformed the static models. 
As the applied approximative baseline model only expresses the dominant high-order dynamics of the real system, all dynamic augmentation structures have outperformed the static models. Introducing the augmented states in the LFR-based structure increases the model DoF compared to static structures. This helps explain why the highest accuracy for the dynamic LFR augmentation is obtained when certain constraints are imposed on the LFR matrix. This example highlights the importance of selecting the optimal model complexity: richer parametrisations can improve expressiveness, but overly flexible models may suffer from reduced robustness. Introducing suitable restrictions can lead to improved performance by mitigating variance effects. It is also worth noting that multiple dynamic LFR-based augmented models have resulted in better model accuracies than the black-box methods.

\begin{table}
    \centering
    \caption{Normalised root mean squared simulation error of the estimated models on the F1Tenth identification study.}
    \begin{tabular}{c||c}
    Model \!&\! Test NRMS error \\
    \hline
     static LFR-based ($D_{zw}\equiv 0$) & 10.71\% \\
     \vspace*{3pt}
     static LFR-based ($D_{zw}^\mathrm{ab}$ tuned) & {\bf 9.27\%} \\
    \vspace*{3pt}
     static LFR-based ($D_{zw}^\mathrm{ba}$ tuned) & 10.44\% \\
    static LFR-based ($D_{zw}\equiv 0$, $z_\mathrm{b}$ fixed) & 9.42\% \\
     \vspace*{3pt}
     static LFR-based ($D_{zw}^\mathrm{ab}$ tuned, $z_\mathrm{b}$ fixed) & 9.90\% \\
    \hline
    dynamic LFR-based ($D_{zw}\equiv 0$) & 8.79\% \\
    \vspace*{3pt}
    dynamic LFR-based ($D_{zw}^\mathrm{ab}$ tuned) & 8.52\% \\
    \vspace*{3pt}
     dynamic LFR-based ($D_{zw}^\mathrm{ba}$ tuned) & 8.99\% \\
    dyn. LFR-based ($D_{zw}\equiv 0$, $z_\mathrm{b}$ fixed) & {\bf 8.25\%} \\
     \vspace*{3pt}
     dyn. LFR-based ($D_{zw}^\mathrm{ab}$ tuned, $z_\mathrm{b}$ fixed) & 8.41\% \\
    \hline
    Initial baseline model & 49.12\%\\
    \vspace*{3pt}
    DT SUBNET (black-box, \cite{szecsi_deep_2024}) & 8.53\%\\
    \vspace*{3pt}
    CT SUBNET (black-box, \cite{szecsi_deep_2024}) & 8.99\%
    \end{tabular}
    \label{tab:f1tenth_results}
\end{table}

\begin{figure}
    \hspace*{-5pt}
    \centering
    \includegraphics[width=\columnwidth]{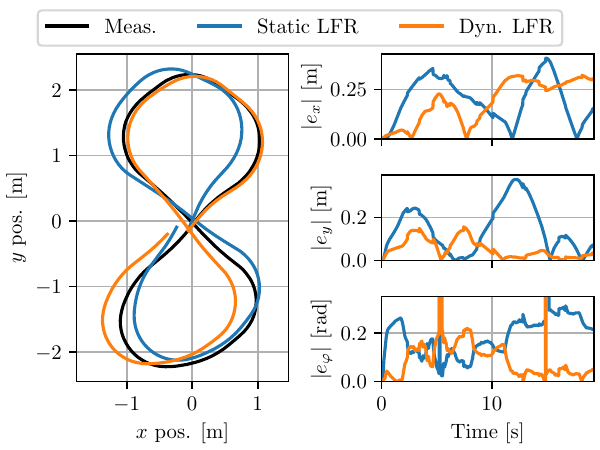}
    \vspace*{-10pt}
    \caption{Comparing the simulated model response with measurements on the test data.}
    \label{fig:f1tenth_results_comparison} 
\vspace{-1.5mm}
\end{figure}

%The convergence of the augmentation structures that achieved the best model accuracy (i.e., when $D_{zw}^\mathrm{ab}$ is tuned for the static case, and when $D_{zw}\equiv 0$, $z_\mathrm{b}$ is fixed for the dynamics case) is compared to the convergence of the DT SUBNET approach (a state-of-the-art black-box identification method) in Fig.~\ref{fig:f1tenth_convergence}. As visible, all three methods have shown similar convergence rates, however, the dynamic LFR-based method performed the best in terms of convergence speed. 
All model augmentation structures have shown similar convergence properties as the DT SUBNET approach (a state-of-the-art black-box identification method); however, the dynamic LFR-based method performed the best in terms of convergence speed. Hence, the proposed model augmentation structure was able to generate more accurate results with better convergence properties than black-box methods.

% \begin{figure}
%     \hspace*{-5pt}
%     \centering
%     \includegraphics[width=\columnwidth]{Figures/convergence.pdf}
%     \vspace*{-15pt}
%     \caption{Validation loss over first 1000 training epochs for the static LFR-based ($D_{zw}^\mathrm{ab}$ tuned) and dynamic LFR-based  ($D_{zw}\equiv 0$, $z_\mathrm{b}$ fixed) structures compared to the convergence of a state-of-the-art black-box identification method.}
%     \label{fig:f1tenth_convergence}
% \end{figure}

\vspace{-1.5mm}\section{Conclusion}\label{sec:Conclusion}
\vspace{-1.5mm}
In this paper a novel general LFR-based model augmentation has been proposed that provides unified representation structure. The model was also expressed in graph form, offering insight into the structural patterns that characterise and enable the detection of specific augmentation structures. In addition, we established conditions ensuring the well-posedness of the proposed model structure.
% Finally, we adapted an encoder-based estimation method for the novel augmentation structure with consistency guarantees.
To provide reliable estimation of the proposed model structure, an adaptation of the SUBNET approach was implemented, inheriting the consistency guarantees of the SUBNET approach. By analysing various augmentation configurations with simulation and  experimental data, we have shown that suitable positioning of the learning component provides faster convergence and can lead to more accurate models compared to state-of-the-art black-box approaches.

\vspace{-1.5mm}
\begin{ack}                               % Place acknowledgements
% This work is funded by the European Union (Horizon Europe, ERC, COMPLETE, 101075836) and has also been supported by the European Defence Fund programme under grant agreement number No 101103386 and has also been supported by the Air Force Office of Scientific Research under award number FA8655-23-1-7061. Views and opinions expressed are however those of the authors only and do not necessarily reflect those of the European Union or the European Research Council Executive Agency or the European Commission. Neither the European Union nor the granting authority can be held responsible for them.
%
%This project has received funding from the European Defence Fund programme under grant agreement number No 101103386 and has also been supported by the Air Force Office of Scientific Research under award number FA8655-23-1-7061. Views and opinions expressed are however those of the authors only and do not necessarily reflect those of the European Union or the European Commission. Neither the European Union nor the granting authority can be held responsible for them.  % here.
\vspace{-1.5mm}
We thank Mircea Lazar, Max Bolderman and Chris Verhoek for the helpful discussions %that were useful
for this work.
\end{ack}

\bibliographystyle{myplain}        % Include this if you use bibtex 
\bibliography{references}           % and a bib file to produce the 
                                 % bibliography (preferred). The
                                 % correct style is generated by
                                 % Elsevier at the time of printing.

\appendix
\vspace{-1mm} \section{Proof of Theorem~\ref{thm:representation_ma_structures}}\label{appendix:general_LFR_augmentation_forms} \vspace{-1mm}
% We show the derivation for the model augmentation structures in Table~\ref{tab:augmentation_structures}. The output augmentations follow in similar form as the state augmentation but instead considering the output mapping matrices. As such we leave out the output mapping here for ease of notation.
%
% \refo{
% Similarly for the output map, $C_y$, $D_{yu}$, and $D_{yw}^\mathrm{a}$ are initialised as zero matrices, moreover, $D_{yw}^\mathrm{b}$ as
% \begin{equation}
%     D_{yw}^{\mathrm{b},0} = \begin{bmatrix}
%         0_{n_y\times n_{\xbase}} & I_{n_y}
%     \end{bmatrix}.
% \end{equation}
% %
% }
%
% \OL{We assume  $n_{z_\mathrm{a}}$ and $n_{w_\mathrm{a}}$ to be of appropriate dimensions to facilitate the defined $z$ and $w$}
% $n_{z_\mathrm{a}}=n_{\xbase}+n_u+n_{\xaug}$ and $n_{w_\mathrm{a}}=n_{\xbase}+n_{\xaug}$
We provide the proof by parts:

\textit{Parallel augmentation at the state level:} Choose  $n_{z_\mathrm{a}}=n_{\xbase}+n_u+n_{\xaug}$ and $n_{w_\mathrm{a}}=n_{\xbase}+n_{\xaug}$. Take $\zbasek \equiv \mathrm{vec}\left(\xbasek, u_k\right)$ and $\zaugk \equiv \mathrm{vec}\left(\xbasek, \xaugk, u_k\right)$. For $\zbasek$, this is achieved by setting \vspace{-1.5mm}
\begin{equation}
    C_z^{\mathrm{b}} = \begin{bmatrix}
        I_{n_{\xbase}} & 0_{n_{\xbase}\times n_{\xaug}}\\
        0_{n_u\times n_{\xbase}} & 0_{n_{\xbase}\times n_{\xaug}}
    \end{bmatrix},\quad
    D_{zu}^{\mathrm{b}} = \begin{bmatrix}
        0_{n_{\xbase}\times n_u}\\ I_{n_u}
    \end{bmatrix},\label{eq:Cz_b0} \vspace{-1.5mm}
\end{equation}
while $D_{zw}^\mathrm{bb}$, and  $D_{zw}^\mathrm{ba}$ are set to zero, while for $\zaugk$, this is achieved by setting \vspace{-3.5mm}
\begin{equation*}
    C_z^\mathrm{a} = \begin{bmatrix}
        I_{n_{\xbase}+n_{\xaug}}^\top&
        0_{n_u\times (n_{\xbase}+n_{\xaug})}^\top
    \end{bmatrix}^\top,\quad
    D_{zu}^\mathrm{a} = \begin{bmatrix}
        0_{(n_{\xbase}+n_{\xaug})\times n_u}^\top&
        I_{n_u}^\top
    \end{bmatrix}^\top, \vspace{-2.5mm}
\end{equation*}
and setting $D_{zw}^\mathrm{ab}$, and $D_{zw}^\mathrm{aa}$ as zeros. Next, we take $\hat x_{k+1} \equiv \wbasek + \waugk$. This is achieved by setting $B_w^{\mathrm{b}} = I_{n_{\xbase}}$, 
% \begin{equation}
%     B_w^{\mathrm{b}} = \begin{bmatrix}
%         I_{n_{\xbase}}
%     \end{bmatrix},
% \end{equation}
and setting $A^\mathrm{bb}$, $A^\mathrm{ba}$, $B_u^\mathrm{b}$, and $B_w^\mathrm{ba}$ as zeros. Then, we have \vspace{-3.5mm}
\begin{equation}
    \begin{bmatrix}
        x_{\mathrm{b}, k+1}\\ x_{\mathrm{a}, k+1}
    \end{bmatrix} = \begin{bmatrix}
        f_\mathrm{base}(\xbasek, u_k)\\ 0_{n_{\xaug}\times n_{\xaug}}
    \end{bmatrix} + \phi_\mathrm{aug}(\xbasek, \xaugk,u_k), \vspace{-2.5mm}
\end{equation}
which is equivalent to the dynamic parallel state augmentation structure. Moreover, since the dynamic parallel structure 
is a generalisation of the static parallel, selecting $n_{\xaug}=0$ results in the static parallel state augmentation structure \vspace{-2.5mm}
\begin{equation}
    x_{\mathrm{b}, k+1} = f_\mathrm{base}(\xbasek, u_k) + \phi_\mathrm{aug}(\xbasek,u_k). \vspace{-1.5mm}
\end{equation}

\textit{Series output augmentation at the state level:} Choose $n_{z_\mathrm{a}}=2n_{\xbase}+n_{\xaug}+n_u$, and $n_{w_\mathrm{a}}=n_{\xbase}+n_{\xaug}$. Take $\zbasek \equiv \mathrm{vec}\left(\xbasek, u_k\right)$, and $\zaugk \equiv \mathrm{vec}\left(\xbasek, \xaugk, u_k, f_\mathrm{base}(\xbasek, u_k)\right)$. The former can be achieved by setting $C_z^\mathrm{b}$, $D_{zu}^\mathrm{b}$ as in \eqref{eq:Cz_b0}, and $D_{zw}^\mathrm{bb}$, $D_{zw}^\mathrm{ba}$ as zeros. The latter can be realised by restricting \vspace{-3.5mm}
\begin{align}
    C_z^\mathrm{a} &= \begin{bmatrix}
        I_{n_{\xbase}+n_{\xaug}}^\top & 0_{(n_u+n_{\xbase})\times(n_{\xbase}+n_{\xaug})}^\top
    \end{bmatrix}^\top,\\
    D_{zu}^\mathrm{a} &= \begin{bmatrix}
        0_{(n_{\xbase}+n_{\xaug})\times n_u}^\top &
        I_{n_u}^\top &
        0_{n_{\xbase}\times n_u}^\top
    \end{bmatrix}^\top,\\
    D_{zw}^\mathrm{ab} &= \begin{bmatrix}
        0_{(n_{\xbase}+n_{\xaug}+n_u)\times n_{\xbase}}^\top &
        I_{n_{\xbase}^\top}
    \end{bmatrix}^\top, 
\end{align} \vskip -4mm
and selecting $D_{zw}^\mathrm{aa}$ as a zero matrix. Then, by selecting $A$, $B_u$, and $B_w^\mathrm{b}$ as zero matrices and setting $B_w^\mathrm{a}$ an identity matrix, the resulting state transition function is \vspace{-1.5mm}
\begin{equation}
    \hat{x}_{k+1} = \phi_\mathrm{aug}(\xbasek, \xaugk, u_k, f_\mathrm{base}(\xbasek, u_k)). \vspace{-1.5mm}
\end{equation}
By universal approximation properties, there exist such weights and biases for $\phi_\mathrm{aug} = \begin{bmatrix}
    (\phi_\mathrm{aug}^\mathrm{b})^\top & (\phi_\mathrm{aug}^\mathrm{a})^\top
\end{bmatrix}^\top$ such that \vspace{-3.5mm}
\begin{equation}
    \begin{bmatrix}
        x_{\mathrm{b}, k+1}\\ x_{\mathrm{a}, k+1}
    \end{bmatrix} = \begin{bmatrix}
        \phi_\mathrm{aug}^\mathrm{b}(\xbasek, \xaugk, u_k, f_\mathrm{base}(\xbasek, u_k))\\
        \phi_\mathrm{aug}^\mathrm{a}(\xbasek, \xaugk, u_k)
    \end{bmatrix}, \vspace{-2.5mm}
\end{equation}
which is equivalent to the dynamic series output state augmentation. Moreover, as we have shown previously, dynamic augmentation is a generalisation of static augmentation; hence, the structure also represents the static series output augmentation form at the state level.

\textit{Static series input augmentation at the state level:} Choose $n_{\zaug} = n_{\xbase}+n_u$ and $n_{\xaug}=0$. Take $\zaugk \equiv \mathrm{vec}\left(\xbasek, u_k, \right)$ by setting $C_z^\mathrm{a}$, $D_{zu}^\mathrm{a}$, $D_{zw}^\mathrm{ab}$,$D_{zw}^\mathrm{bb}$, $A$, $B_u$, $B_w^\mathrm{b}$ and $D_{zw}^\mathrm{aa}$ as for the \emph{parallel augmentation at the state level} structure. Set $B_w^\mathrm{a}\equiv 0$, $C_z^\mathrm{b}\equiv 0$, $D_{zu}^\mathrm{b}\equiv 0$ and $D_{zw}^\mathrm{ba}$ as an identity matrix. This achieves \vspace{-2.5mm} %a state transition as
\begin{align}
    \zbasek &= \phi_\mathrm{aug}(\xbasek, u_k),\\
    x_{\mathrm{b},k+1} &= f_\mathrm{base}(\zbasek), \vspace{-3.5mm}
\end{align}
\vskip -4.5mm
which is equivalent to the static series input augmentation form at the state transition level.

\textit{Dynamic series input augmentation on the state level:} Choose $n_{\zbase} = n_{\xbase}+n_{\xaug}+n_u$. Take $\zaugk \equiv \mathrm{vec}\left(\xbasek, \xaugk, u_k, \right)$ by setting $C_z^\mathrm{a}$, $D_{zu}^\mathrm{a}$, $D_{zw}^\mathrm{ab}$,$D_{zw}^\mathrm{bb}$, $A$, $B_u$ and $D_{zw}^\mathrm{aa}$ as for the \emph{parallel augmentation at the state level} structure. Set $C_z^\mathrm{b}\equiv 0$, $D_{zu}^\mathrm{b}\equiv 0$ and set $D_{zw}^\mathrm{ba}$ as an identity matrix. Consider the output of the learning component to be split into two parts $\phi_\mathrm{aug} = \begin{bmatrix}
    (\phi_\mathrm{aug}^\mathrm{b})^\top & (\phi_\mathrm{aug}^\mathrm{a})^\top
\end{bmatrix}^\top$.
% \OL{
% Select $D_{zw}^\mathrm{bb}$, $A$ and $B_u$ according to Appendix~\ref{appendix:parallel_state_augm_param}. Then , $C_z^\mathrm{b}\equiv 0$ and $D_{zu}^\mathrm{b}\equiv 0$ then set $D_{zw}^\mathrm{ba}$ as an identity matrix.
% $B_w^\mathrm{b} = \begin{bmatrix}
% \end{bmatrix}$ $B_w^\mathrm{a}\equiv 0$
% }
This achieves \vspace{-3.5mm} %a state transition as
\begin{equation}
    \begin{bmatrix}
        x_{\mathrm{b}, k+1}\\ x_{\mathrm{a}, k+1}
    \end{bmatrix} = \begin{bmatrix}
        x_{\mathrm{b},k+1} = f_\mathrm{base}(\phi_\mathrm{aug}^\mathrm{b}(\xbasek, \xaugk, u_k))\\
        \phi_\mathrm{aug}^\mathrm{a}(\xbasek, \xaugk, u_k)
    \end{bmatrix}, \vspace{-2.5mm}
\end{equation}
which is equivalent to the dynamic series input augmentation form at the state transition level.

\textit{Output augmentation structures:} The output augmentation formulations are similar in structure to the state augmentation cases, but the connection of the baseline and learning components happens at the output level. Following the previous arguments, it is straightforward to derive that all structures in Table~\ref{tab:output_augmentation_structures} can be represented by \eqref{eqs:general_LFR}, similarly to the structures in Table~\ref{tab:augmentation_structures}, which concludes the proof.

\end{document}